\documentclass[a4paper,11pt,oneside]{article}

\usepackage[english]{babel}
\usepackage[mathscr]{eucal}
\usepackage{amsmath}
\usepackage{mathrsfs}
\usepackage{fullpage}
\usepackage{amsthm}
\usepackage{amsfonts}
\usepackage{amssymb}
\usepackage{amscd}
\usepackage{bbm}
\usepackage[affil-it]{authblk}

\newcommand{\ud}{\mathrm{d}}
\newcommand{\ii}{\mathrm{i}}

\newcommand{\de}{\partial}

\newcommand{\cH}{\mathcal{H}}

\newtheorem{theorem}{Theorem}[section]
\newtheorem{definition}{Definition}[section]
\newtheorem{remark}{Remark}[section]
\newtheorem{lemma}{Lemma}[section]
\newtheorem{proposition}{Proposition}[section]

\title{Gross-Pitaevskii non-linear dynamics for pseudo-spinor condensates\\
}

\author{\footnotesize Alessandro Michelangeli\thanks{alemiche@sissa.it}
}

\affil{SISSA -- International School for Advanced Studies\\ 
	Via Bonomea 265, 34136 Trieste (Italy)
	\\
}

\author{\footnotesize Alessandro Olgiati\thanks{aolgiati@sissa.it
	}
}

\affil{SISSA -- International School for Advanced Studies\\ 
	Via Bonomea 265, 34136 Trieste (Italy)
	\\
}

\begin{document}

\maketitle
\thispagestyle{empty}

\begin{abstract}
We derive the equations for the non-linear effective dynamics of a so called pseudo-spinor Bose-Einstein condensate, which emerges from the linear many-body Schr\"{o}dinger equation at the leading order in the number of particles. The considered system is a three-dimensional diluted gas of identical bosons with spin, possibly confined in space, and coupled with an external time-dependent magnetic field; particles also interact among themselves through a short-scale repulsive interaction. The limit of infinitely many particles is monitored in the physically relevant Gross-Pitaevskii scaling. In our main theorem, if at time zero the system is in a phase of complete condensation (at the level of the reduced one-body marginal) and with energy per particle fixed by the Gross-Pitaevskii functional, then such conditions persist also at later times, with the one-body orbital of the condensate evolving according to a system of non-linear cubic Schr\"{o}dinger equations coupled among themselves through linear (Rabi) terms. The proof relies on an adaptation to the spinor setting of Pickl's projection counting method developed for the scalar case. Quantitative rates of convergence are available, but not made explicit because evidently non-optimal. In order to substantiate the formalism and the assumptions made in the main theorem, in an introductory section we review the mathematical formalisation of modern typical experiments with pseudo-spinor condensates.\vspace{0.2cm}\vspace{0.2cm}
\end{abstract}

\section{Introduction: pseudo-spinor Bose-Einstein condensation}\label{sec:intro}

It is customary to refer to pseudo-spinor condensates as gases of ultra-cold atoms that exhibit a macroscopic occupation of the same one-body state (Bose-Einstein condensation) and possess internal spin degrees of freedom which are often coupled to an external resonant micro-wave or radio-frequency radiation field, however, with no significant spin-spin internal interaction (whence the \emph{pseudo}-spinor terminology). 
The order parameter of the condensation is therefore a multi-component vector, unlike scalar condensates such as liquid $^4\mathrm{He}$, and the dynamical evolution of these quantum fluids observed in the experiments shows an excellent matching with a non-linear effective dynamics for the order parameter.

In this work we want to present a rigorous derivation of such non-linear equations from the `first principle' many-body linear Schr\"{o}dinger dynamics.

The study of multi-component Bose-Einstein condensates (henceforth also BEC) was spurred on in 1997-1998 by experiments on ultra-cold Rubidium, with condensation coexisting in two different hyperfine states of $^{87}\mathrm{Rb}$ \cite{MBGCW-1997,Matthews_HJEWC_DMStringari_PRL1998,HMEWC-1998,Hall-Matthews-Wieman-Cornell_PRL81-1543} and soon extended to multi-BEC for heteronuclear mixtures such as $^{41}\mathrm{K}$-$^{87}\mathrm{Rb}$ \cite{Modugno-Ferrari-Inguscio-etal-Science2001_multicompBEC}, $^{41}\mathrm{K}$-$^{85}\mathrm{Rb}$ \cite{Modugno-PRL-2002}, $^{39}\mathrm{K}$-$^{85}\mathrm{Rb}$ \cite{MTCBM-PRL2004_BEC_heteronuclear}, $^{85}\mathrm{Rb}$-$^{87}\mathrm{Rb}$ \cite{Papp-Wieman_PRL2006_heteronuclear_RbRb}. In the last two decades the field has expanded through a huge amount of experimental and theoretical studies, for a survey of which we refer to the comprehensive reviews \cite{Ketterle_StamperKurn_SpinorBEC_LesHouches2001, Malomed2008_multicompBECtheory, Hall2008_multicompBEC_experiments,StamperKurn-Ueda_SpinorBose_Gases_2012} 
(see also \cite[Chapter 21]{pita-stringa-2016}).


In order to place the present work into the appropriate mathematical setting, it is instructive to revisit, within the general formalism of many-body quantum mechanics, the essential steps of a typical experiment -- for concreteness we refer to the 1998 pioneering experiment \cite{Matthews_HJEWC_DMStringari_PRL1998,HMEWC-1998,Hall-Matthews-Wieman-Cornell_PRL81-1543}.

First and foremost, the experiment involves only a few hyperfine levels of the considered atomic species -- for $^{87}\mathrm{Rb}$ these are the $5S_{1/2}|F=1,m_f=-1\rangle$ and $5S_{1/2}|F=2,m_f=1\rangle$ states: this results in the effective one-body Hilbert space
\begin{equation}\label{eq:h-one-body}
\mathfrak{h}\;:=\;L^2(\mathbb{R}^3)\otimes\mathbb{C}^2\;\cong\;L^2(\mathbb{R}^3)\oplus L^2(\mathbb{R}^3)
\end{equation}
as for one spin-$\frac{1}{2}$ particle in three spatial dimensions. (However, for the final measurement process the effective Hilbert space to consider is a larger one, as we shall explain later.) The corresponding many-body bosonic Hilbert space is
\begin{equation}\label{eq:H-many-body}
\cH_N\;:=\;\mathfrak{h}^{\,\otimes_\mathrm{sym} N}\,,
\end{equation}
the  \emph{symmetric} $N$-fold tensor product of $\mathfrak{h}$. Elements of $\mathfrak{h}$ are spinors $\begin{pmatrix} u_\uparrow \\ u_\downarrow\end{pmatrix}$ with $u_\uparrow,u_\downarrow\in L^2(\mathbb{R}^3)$, equivalently, $u\cdot\begin{pmatrix} c_1 \\ c_2\end{pmatrix}$ with $u\in L^2(\mathbb{R}^3)$, $c_1,c_2\in\mathbb{C}$. With reference to the two actual hyperfine levels entering the experiment, we denote $|1,-1\rangle\equiv|\uparrow\rangle\equiv\begin{pmatrix} 1 \\ 0\end{pmatrix}$ and $|2,1\rangle\equiv|\downarrow\rangle\equiv\begin{pmatrix} 0 \\ 1\end{pmatrix}$.

Through a very ingenious confining and cooling procedure, $N\sim 10^5$ atoms are prepared inside an optical trap and brought to complete condensation onto the one-body state $\begin{pmatrix} u_0 \\ 0\end{pmatrix}$. The experimental evidence is that \emph{no noticeable non-condensed fraction remains}, thus implying that the \emph{c}onfined and \emph{c}ooled many-body state $\Psi_{N,\mathrm{cc}}\in\cH_N$ displays a 100\% macroscopic occupation of the orbital $\begin{pmatrix} u_0 \\ 0\end{pmatrix}$. Ideally this would mean that 
\begin{equation}\label{eq:initial0u}
\Psi_{N,\mathrm{cc}}\;\sim\;\begin{pmatrix} u_0 \\ 0\end{pmatrix}^{\!\!\otimes N}\qquad (\,\|u_0\|_2\;=\;1\,)
\end{equation}
holds as an identity in $\cH_N$, or at least in a thermodynamic limit $N\to \infty$. 
However, as customary in the mathematical formalisation of complete BEC \cite{LSeSY-ober,am_equivalentBEC}, it is more appropriate to infer the meaning of occupation numbers from the eigenvalues of the one-body reduced density matrix associated with the many-body state -- for otherwise even the negligible change of one single one-body orbital out of $N$ in the state $\begin{pmatrix} u_0 \\ 0\end{pmatrix}^{\!\!\!\!\otimes N}$ would result in a new state that would be essentially orthogonal to the original one.

Let us recall that, associated to each $\Psi_{N}\in\cH_N$, or more generally to each many-body density matrix $\gamma_N$ on $\cH_N$,  is the so-called one-body marginal (or one-body reduced density matrix)
\begin{equation}\label{eq:partial_trace1}
\gamma_N^{(1)}\;=\;\mathrm{Tr}_{N-1}\,\gamma_N\,,
\end{equation}
where the map $\mathrm{Tr}_{N-1}:\mathcal{B}_1(\cH_N)\to\mathcal{B}_1(\mathfrak{h})$ is the \emph{partial trace} from trace class operators acting on $\cH_N$ to trace class operators acting on $\mathfrak{h}$. $\mathrm{Tr}_{N-1}\,\gamma_N$ is defined, by duality, by
\begin{equation}\label{eq:def_partial_trace_without_basis}
 \mathrm{Tr}_\mathfrak{h}(A\cdot\mathrm{Tr}_{N-1}\,\gamma_N)\;=\;\mathrm{Tr}_{\cH}(A\otimes\mathbbm{1}_{N-1})\cdot \gamma_N))\qquad\forall A\in\mathcal{B}(\mathfrak{h})
\end{equation}
(here $\mathcal{B}$ denotes the bounded linear operators on $\mathfrak{h}$ and $\mathcal{B}_1$ the corresponding trace class).
In terms of an arbitrary orthonormal basis $(\Xi_k)_k$ of $\cH_{N-1,\mathrm{sym}}$ one then has
\begin{equation}\label{eq:def_partial_trace_with_basis}
\langle \varphi,(\mathrm{Tr}_{N-1}\,\gamma_N)\psi\rangle_\mathfrak{h}\;=\;\sum_{k}\langle\varphi\otimes\Xi_k,\gamma_N\,\psi\otimes\Xi_k\rangle_{\cH_{N-1,\mathrm{sym}}}\qquad\forall\varphi,\psi\in\mathfrak{h}\,,
\end{equation}
and the l.h.s.~of \eqref{eq:def_partial_trace_with_basis} is independent of the choice of the basis.
Thus, $\gamma_N^{(1)}$ is obtained by ``tracing out'' $N-1$ degrees of freedom from $\gamma_N$.
As a non-negative, bounded, and self-adjoint operator on $\mathfrak{h}$,  $\gamma_N^{(1)}$ has a complete set of real non-negative eigenvalues that sum up to 1, that is, there is an orthonormal basis $(\varphi_j^{(N)})_{j=0}^\infty$ of $\mathfrak{h}$ consisting of eigenvectors of $\gamma_N^{(1)}$ with eigenvalues $(n_j^{(N)})_{j=0}^\infty$\, so that 
\begin{equation}
\begin{array}{c}
\gamma_N^{(1)}\;=\;\sum_{j=0}^\infty\,n_j^{(N)}\,|\varphi_j^{(N)}\rangle\langle\varphi_j^{(N)}|\,, \\
1\geqslant n_0^{(N)}\geqslant n_1^{(N)}\geqslant\cdots\geqslant 0\,,\qquad\sum_{j=0}^\infty n_j^{(N)}=1\,.
\end{array}
\end{equation}
Thanks to the bosonic symmetry, each such eigenvalue has the natural interpretation of \emph{occupation number}: indeed, since the one-body observable (given by the global symmetrization of) $\mathcal{O}_j\equiv|\varphi_j^{(N)}\rangle\langle\varphi_j^{(N)}|\otimes\mathbbm{1}_{N-1}$ has expectation $n_j^{(N)}$ in the many-body state $\gamma_N$, as follows from (\ref{eq:def_partial_trace_without_basis}), then $n_j^{(N)}$ expresses \emph{in the sense of the reduced density matrix}, the fraction of particles of the many-body state $\gamma_N$ which occupy the one-body state $\varphi_j^{(N)}$. Complete BEC of the many-body state $\Psi_N$ onto the one-body orbital $\varphi_0\in\mathfrak{h}$ is by definition the occurrence $n_0^{(N)}\sim 1$, $\varphi_0^{(N)}\sim\varphi_0$ and $\gamma_N^{(1)}\sim |\varphi_0\rangle\langle\varphi_0|$, where an underlying thermodynamic limit $N\to\infty$ is tacitly assumed. This is so in the ideal case of a completely factorised $\Psi_N=\varphi_0^{\otimes N}$, however $\gamma_N^{(1)}\sim |\varphi_0\rangle\langle\varphi_0|$ is a much weaker statement, since an amount of correlations that are not negligible in the many-body norm may be present in $\Psi_N$.

Thus, the many-body state $\Psi_{N,\mathrm{cc}}$ after the initial confinement and cooling is prepared as in \eqref{eq:initial0u} in the sense of the reduced density matrix, that is,
\begin{equation}
\gamma_{N,\mathrm{cc}}^{(1)}\;\approx\;\Big|\!\begin{pmatrix} u_0 \\ 0\end{pmatrix}\!\Big\rangle\Big\langle\!\begin{pmatrix} u_0 \\ 0\end{pmatrix}\!\Big|\;=\;\begin{pmatrix} |u_0\rangle\langle u_0|\;\; & \mathbb{O} \\ \mathbb{O}\;\; & \mathbb{O}\end{pmatrix}.
\end{equation}
Here $u_0$ is the minimiser of a suitable energy functional, which corresponds to the fact that $\Psi_{N,\mathrm{cc}}$ is a ground state, in the sector of `all spin up particles', of the (effective) many-body Hamiltonian
\begin{equation}\label{eq:HNprep}
\sum_{j=1}^N\big(-{\textstyle\frac{\:\hbar^2}{\,2m\,}}\,\Delta_j+U^\mathrm{trap}(x_j)\big)\;+\!\sum_{1\leqslant j<k\leqslant N}V(x_j-x_k)\,,
\end{equation}
where $V$ is the potential for the two-body interaction that depends only on the particle spatial configuration, and
\begin{equation}
U^\mathrm{trap}(x)\;=\;\;\begin{pmatrix} U^{\mathrm{trap}}_\uparrow(x) & 0 \\ 0 & U^{\mathrm{trap}}_\downarrow(x)\end{pmatrix}
\end{equation}
models the (typically harmonic) external trapping potential. Since $\Psi_{N,\mathrm{cc}}$ consists (essentially) only of `spin up' particles, it is only subject to the confining potential $U^{\mathrm{trap}}_\uparrow$. In fact, in the experiment $U^{\mathrm{trap}}_\uparrow(x)\approx U^{\mathrm{trap}}_\downarrow(x)$ to within $\sim$0.3\%.

The next step of the experiment is a `two-photon transition', consisting of a very quick ($\sim400\mu\mathrm{s}$) pulse of an external oscillating radiofrequency field that couples with the spin of the particles and is tuned close to the hyperfine splitting energy $2V_{\mathrm{hf}}$ between the two levels, so as to connect the $|1,-1\rangle$ state to the $|2,1\rangle$, with an action on each spinor which in the `rotating wave approximation' is generated by
\begin{equation}\label{eq:2ph}
 \ii\hbar\partial_t\begin{pmatrix} \psi_1 \\ \psi_2\end{pmatrix}\;=\;\begin{pmatrix} -V_{\mathrm{hf}} & \hbar\Omega \,e^{\ii\hbar\omega t} \\ \hbar\Omega \,e^{-\ii\hbar\omega t} & V_{\mathrm{hf}}\end{pmatrix}
 \begin{pmatrix} \psi_1 \\ \psi_2\end{pmatrix},\quad V_{\mathrm{hf}}\;=\;{\textstyle\frac{1}{2}}\hbar\omega\,,\qquad \begin{pmatrix} \psi_1(0) \\ \psi_2(0)\end{pmatrix}\;=\;\begin{pmatrix} u_0 \\ 0\end{pmatrix},
\end{equation}
with $\Omega\sim 2\pi\cdot 625\,\mathrm{Hz}$ and $\omega\sim 2\pi\cdot 6.8\,\mathrm{GHz}$. (Thus, $\omega\gg\Omega$, as appropriate for the rotating wave approximation.) The evolution governed by \eqref{eq:2ph} involves only the spin degrees of freedom, based on the fact that the duration of the applied pulse is much shorter than the characteristic time for the spatial wave function of each spinor to change appreciably. That this is actually so as a consequence of first principles must of course be demonstrated. With the transformation
\begin{equation}
\begin{pmatrix} \widetilde{\psi}_1 \\ \widetilde{\psi}_2\end{pmatrix}\;:=\;
\begin{pmatrix} e^{-\frac{1}{2}\ii\hbar\omega t} & 0 \\ 0 & e^{\frac{1}{2}\ii\hbar\omega t}\end{pmatrix}
\begin{pmatrix} \psi_1 \\ \psi_2\end{pmatrix}
\end{equation}
\eqref{eq:2ph} reads
\begin{equation}\label{eq:2ph2}
 \ii\hbar\partial_t\begin{pmatrix} \widetilde{\psi}_1 \\ \widetilde{\psi}_2\end{pmatrix}\;=\;\begin{pmatrix} 0 & \hbar\Omega\\ \hbar\Omega  & 0\end{pmatrix}
 \begin{pmatrix} \widetilde{\psi}_1 \\ \widetilde{\psi}_2\end{pmatrix}\,,\qquad \begin{pmatrix} \widetilde{\psi}_1(0) \\ \widetilde{\psi}_2(0)\end{pmatrix}\;=\;\begin{pmatrix} u_0 \\ 0\end{pmatrix},
\end{equation}
whence
\begin{equation}
\begin{pmatrix} e^{-\frac{1}{2}\ii\hbar\omega t}\psi_1(t) \\ e^{\frac{1}{2}\ii\hbar\omega t}\psi_2(t)\end{pmatrix}\;=\;
\begin{pmatrix} \widetilde{\psi}_1(t) \\ \widetilde{\psi}_2(t)\end{pmatrix}\;=\;
\begin{pmatrix} \cos\Omega t & -\ii\,\sin\Omega t \\ -\ii\,\sin\Omega t & \cos\Omega t\end{pmatrix}\begin{pmatrix} \psi_1(0) \\ \psi_2(0)\end{pmatrix}\;=\;\begin{pmatrix} u_0\,\cos\Omega t \\ -\ii\,u_0\,\sin\Omega t\end{pmatrix}.
\end{equation}

It is worth underlying that in the course of the two-photon transition the external field couples simultaneously with the spin of \emph{each} particle and yields then a many-body Schr\"{o}dinger equation that is the product of $N$ copies of \eqref{eq:2ph}). The particles in the (almost) factorised state $\Psi_{N,\mathrm{cc}}$ are therefore all rotated the same way and at the end of this phase, say, at time $t_0$, the many-body state is then transformed into a `\emph{r}otated' one as
\begin{equation}\label{eq:psicc-psir}
\Psi_{N,\mathrm{cc}}\;\longmapsto\;\Psi_{N,\mathrm{r}}\;\sim\;\begin{pmatrix} u_0\,e^{\frac{1}{2}\ii\hbar\omega t_0}\,\cos\Omega t_0 \\ -u_0\,\ii\,e^{-\frac{1}{2}\ii\hbar\omega t_0}\,\sin\Omega t_0\end{pmatrix}^{\!\!\!\otimes N}\,.
\end{equation}
$\Psi_{N,\mathrm{r}}$ still exhibits complete BEC, however on a one-body orbital with the same spatial wave-function and rotated spin.

Whereas the measurement process, that we shall describe later below and is a destructive procedure, may take place already at this stage for the state $\Psi_{N,\mathrm{r}}$, other steps may be also performed in the experiment, before the final measurement. One possible further step is to let $\Psi_{N,\mathrm{r}}$ relax in the trap until it reaches the ground state of the Hamiltonian \eqref{eq:HNprep} -- having been rotated, now $\Psi_{N,\mathrm{r}}$ is not an eigenstate any more for \eqref{eq:HNprep}. It is expected and observed (and need be proved also from first principles) that this relaxation does not alter substantially the almost factorised structure and produces a \emph{r}otated and \emph{r}elaxed many-body state
\begin{equation}\label{eq:psirr}
\Psi_{N,\mathrm{cc}}\;\longmapsto\;\Psi_{N,\mathrm{r}}\;\longmapsto\;\Psi_{N,\mathrm{rr}}\;\sim\;\begin{pmatrix} u_0' \\ v_0'\end{pmatrix}^{\!\!\!\otimes N},\qquad\|u_0'\|_2^2+\|v_0'\|_2^2=1\,.
\end{equation}
Because of the actual experimental values of the trapping and interaction potentials in \eqref{eq:HNprep}, typically $u_0'$ and $v_0'$ are essentially supported on almost disjoint regions of space -- a phenomenon customarily referred to as \emph{phase separation} -- which in particular makes them orthogonal: $u_0'\perp v_0'$.

Another possible further step in the experiment, right after the rotation or the relaxation, consists of switching off the confinement too ($U^{\mathrm{trap}}\equiv 0$) and letting the gas expand hydrodynamically, subject only to the mutual inter-particle interaction, for a period of some $\sim20\,\mathrm{ms}$. 
In the course of this expansion spatial correlations are developed between the $N$ initially factorised spinors: it is an experimental evidence, that one expects to demonstrate also from first principles, that for times that are much smaller then the time scale at which the condensate deteriorates completely the (almost) factorised structure of the many-body state is essentially preserved.  This part of the experiment then produces an \emph{e}xpanded condensate with many-body state 
\begin{equation}\label{eq:psie}
\Psi_{N,\mathrm{cc}}\;\longmapsto\;\Psi_{N,\mathrm{r}}\;(\textrm{or }\Psi_{N,\mathrm{r}})\;\longmapsto\;\Psi_{N,\mathrm{e}}\;\sim\;\begin{pmatrix} u \\ v \end{pmatrix}^{\!\!\!\otimes N},\qquad\|u\|_2^2+\|v\|_2^2=1\,.
\end{equation}

At the end of the rotation, or the relaxation, or the expansion, information is read out of the many-body state  ($\Psi_{N,\mathrm{r}}$, or $\Psi_{N,\mathrm{rr}}$, or $\Psi_{N,\mathrm{e}}$) with a procedure that is destructive, however the excellent reproducibility of the condensate $\Psi_{N,\mathrm{cc}}$ allows one to repeat the measurement for various different times.

This process can be effectively described in a \emph{larger} one-body Hilbert space than $\mathfrak{h}$ given in \eqref{eq:h-one-body}, for a \emph{third} hyperfine level  in the $5 P_{3/2}\;F=3$ manifold is allowed to be reached. One then considers the space
\begin{equation}
\mathfrak{h}'\;:=\;L^2(\mathbb{R}^3)\otimes\mathbb{C}^3\;\cong\;L^2(\mathbb{R}^3)\oplus L^2(\mathbb{R}^3)\oplus L^2(\mathbb{R}^3)
\end{equation}
where the previous two spinors $|\uparrow\rangle$ and $|\downarrow\rangle$ and the third one $|\mathfrak{m}\rangle$ used for the measurement are identified as 
\begin{equation}
|\uparrow\rangle\equiv\begin{pmatrix} 1 \\ 0 \\ 0\end{pmatrix}\,,\qquad |\downarrow\rangle\equiv\begin{pmatrix} 0 \\ 1 \\ 0\end{pmatrix}\,,\qquad |\mathfrak{m}\rangle\equiv\begin{pmatrix} 0 \\ 0 \\ 1\end{pmatrix}\,.
\end{equation}
At this effective level, the possible experimental manipulations in this process are:
\begin{itemize}
 \item  `pumping' action on $\mathfrak{h}'$ 
 (which is implemented by a short pulse of `repump' light), the effect of which is to produce the change 
\[
 \begin{pmatrix} u \\ v \\ 0\end{pmatrix}\;\stackrel{P}{\longmapsto}\;\begin{pmatrix} 0 \\ u+v \\ 0\end{pmatrix};
\]
 \item  `blowing' action on $\mathfrak{h}'$ 
 (which is implemented by a $\sim2\,\mathrm{ms}$, $\sim60\,\mu\mathrm{W}/\mathrm{cm}^2$ pulse of light that brings $|\downarrow\rangle\mapsto|\mathfrak{m}\rangle$ and has no effect on the $|\uparrow\rangle$ atoms, and `blows away' particles in the hyperfine level $|\mathfrak{m}\rangle$ to a region far from the imaging region, hence practically out of the system), the effect of which is to produce the change 
\[
 \begin{pmatrix} u \\ v \\ 0\end{pmatrix}\;\stackrel{B}{\longmapsto}\;\begin{pmatrix} u \\ 0 \\ 0\end{pmatrix};
\]
 \item `probing' action on $\mathfrak{h}'$ (which is implemented with a $\sigma^+$ circularly polarised probe beam at $\sim17$ MHz that brings $|\downarrow\rangle\mapsto|\mathfrak{m}\rangle$, while atoms in the $|\uparrow\rangle$ state are far (6.8 GHz) from resonance and invisible to the probe beam), the effect of which is to produce the change 
\[
 \begin{pmatrix} u \\ v \\ 0\end{pmatrix}\;\stackrel{C}{\longmapsto}\;\begin{pmatrix} u \\ 0 \\ v\end{pmatrix};
\]
%
%
%
 \item and the actual measurement process, that consists of imaging the shadow of the above-mentioned circularly polarised probe beam onto a charged-coupled device camera (CCD) and hence corresponds first to projecting each spinor orthogonally onto the level $|\mathfrak{m}\rangle$, and then to performing one-body position observations  (this is indeed the set of data that can read out from the CCD) tracing out the spin degrees of freedom: symbolically,
\[
 \begin{pmatrix} u \\ v \\ w\end{pmatrix}\;\stackrel{S}{\longmapsto}\;\begin{pmatrix} 0 \\ 0 \\ w\end{pmatrix}\;\stackrel{}{\longmapsto}\;|w\rangle\langle w|\,.
\]
%
%
\end{itemize}

This way, spatial measurements for the level $|\uparrow\rangle$ are done via the sequence
\begin{equation}
\begin{pmatrix} u \\ v \\ 0\end{pmatrix}\;\stackrel{B}{\longmapsto}\;\begin{pmatrix} u \\ 0 \\ 0\end{pmatrix}\;\stackrel{P}{\longmapsto}\;\begin{pmatrix} 0 \\ u \\ 0\end{pmatrix}\;\stackrel{C}{\longmapsto}\;\begin{pmatrix} 0 \\ 0 \\ u\end{pmatrix}\;\stackrel{S+CCD}{\longmapsto}\;\;|u\rangle\langle u|\,,
\end{equation}
and spatial measurements for the level $|\downarrow\rangle$ are done via a similar sequence that does not include the repump procedure, namely 
\begin{equation}
\begin{pmatrix} u \\ v \\ 0\end{pmatrix}\;\stackrel{C}{\longmapsto}\;\begin{pmatrix} u \\ 0 \\ v\end{pmatrix}
\;\stackrel{S+CCD}{\longmapsto}\;|v\rangle\langle v|\,.
\end{equation}

Since, as we remark once again, all the above-mentioned preparation and measurement procedures involve \emph{simultaneously each spinor} (i.e., it is experimentally impossible to act selectively on some spinors, no multi-body observable of that sort is available), only states of the form $\begin{pmatrix} u \\ v \end{pmatrix}^{\!\!\!\otimes N}$ are manipulated throughout and are accessed to, in the sense of one-body reduced density matrices. For example, no initial state of the form
\begin{equation}
\Big(\begin{pmatrix} w_1 \\ 0 \end{pmatrix}^{\!\!\!\otimes N_1}\!\!\otimes\begin{pmatrix} 0 \\ w_2 \end{pmatrix}^{\!\!\!\otimes N_2}\Big)_{\mathrm{sym}}\qquad (\|w_1\|_2=\|w_2\|_2=1,\,N_1+N_2=N)
\end{equation}
is preparable in this context. Let us then survey what one-body interpretation is possible for the $N$-body states of interest.

First, as discussed already, a generic (pure) state $\Psi_N\in\mathcal{H}_{N,\mathrm{sym}}$ which is accessible only by one-body observables is effectively described by the one-body marginal \eqref{eq:partial_trace1}, where $N-1$ one-body degrees of freedom have been traced out:
 $\gamma_N^{(1)}\;=\;\mathrm{Tr}_{N-1}\,|\Psi_N\rangle\langle\Psi_N|$. Thus, when we deal with pseudo-spinor condensates we shall be only concerned with the class of marginals of the form 
\begin{equation}
\gamma_N^{(1)}\;\approx\;\Big|\!\begin{pmatrix} u \\ v\end{pmatrix}\!\Big\rangle\Big\langle\!\begin{pmatrix} u \\ v\end{pmatrix}\!\Big|\;=\;\begin{pmatrix} |u\rangle\langle u| & |u\rangle\langle v| \\ |v\rangle\langle u| & |v\rangle\langle v|\end{pmatrix},\qquad \|u\|_2^2+\|v\|_2^2=1
\end{equation}
(asymptotically in $N$).

As $\gamma^{(1)}_N$ is a density matrix acting on the one-body Hilbert space $\mathfrak{h}=L^2(\mathbb{R}^3)\otimes\mathbb{C}^2$, natural observables to be evaluated in $\gamma_N^{(1)}$ are orbital-only (i.e., trivial on the spin sector $\mathbb{C}^2$) or spin-only (i.e., trivial on the orbital sector $L^2(\mathbb{R}^3)$). This is precisely what discussed above concerning the experiments: in the measurement process, the $P,B,C,S$ operations are spin-only, whereas the imaging onto the CCD camera is orbital-only. Thus, the expectation on $\gamma_N^{(1)}$ of the spin-only observable of `spin-up' particle is
\begin{equation}
\mathrm{Tr}_{L^2(\mathbb{R}^3)\otimes\mathbb{C}^2}\Big(\gamma_N^{(1)}\cdot\big(\mathbbm{1}\otimes|\uparrow\rangle\langle\uparrow|\big)\Big)\;=\;\mathrm{Tr}_{\mathbb{C}^2}\big(\gamma_{N,\mathrm{spin}}^{(1)}\,|\uparrow\rangle\langle\uparrow|\big)\;=\;\|u\|_2^2
\end{equation}
and the expectation on $\gamma_N^{(1)}$ of the spin-only observable of `spin-down' particle is
\begin{equation}
\mathrm{Tr}_{L^2(\mathbb{R}^3)\otimes\mathbb{C}^2}\Big(\gamma_N^{(1)}\cdot\big(\mathbbm{1}\otimes|\downarrow\rangle\langle\downarrow|\big)\Big)\;=\;\mathrm{Tr}_{\mathbb{C}^2}\big(\gamma_{N,\mathrm{spin}}^{(1)}\,|\downarrow\rangle\langle\downarrow|\big)\;=\;\|v\|_2^2\,,
\end{equation}
where
\begin{equation}
\gamma_{N,\mathrm{spin}}^{(1)}\;:=\;\mathrm{Tr}_{\mathrm{orb}}\gamma_{N}^{(1)}\;=\;\begin{pmatrix} \|u\|_2^2 & \langle v,u\rangle \\ \langle u,v\rangle & \|v\|_2^2\end{pmatrix}
\end{equation}
is the one-body spin-only reduced density matrix acting on $\mathbb{C}^2$ and the partial trace $\mathrm{Tr}_{\mathrm{orb}}$ traces out the orbital degrees of freedom.

Therefore, the many-body state (pseudo-spinor condensate)  $\begin{pmatrix} u \\ v \end{pmatrix}^{\!\!\!\otimes N}$  is to be interpreted as an assembly of $N$ identical bosons for each of which the probability of occupying the level $|\uparrow\rangle$ is $\|u\|_2^2$ and the probability of occupying the level $|\downarrow\rangle$ is $\|v\|_2^2$. By combining this with the above discussion on the actual experiments, and in the sense discussed so far, we are to think of the state $\begin{pmatrix} u \\ v \end{pmatrix}^{\!\!\!\otimes N}$ as a many-body state of identical spin-$\frac{1}{2}$ bosons, of which $N\cdot\|u\|_2^2$ are the fraction of particles in the level $|\uparrow\rangle$ and (normalised) spatial orbital $\|u\|_2^{-1}u$, and $N\cdot\|v\|_2^2$ are those in the level $|\downarrow\rangle$ and (normalised) spatial orbital $\|v\|_2^{-1}v$.

Joint spatial measurements for both the $|\uparrow\rangle$ and the $|\downarrow\rangle$ level can be performed too, through the sequence
\begin{equation}
\begin{pmatrix} u \\ v \\ 0\end{pmatrix}\;\stackrel{P}{\longmapsto}\;\begin{pmatrix} 0 \\ u+v \\ 0\end{pmatrix}\;\stackrel{C}{\longmapsto}\;\begin{pmatrix} 0 \\ 0 \\ u+v\end{pmatrix}\;\stackrel{S+CCD}{\longmapsto}\;\;|u+v\rangle\langle u+v|\,.
\end{equation}
This is a particularly informative measurement when the orbitals $u$ and $v$ are orthogonal, and hence $\|u+v\|_2^2=\|u\|_2^2+\|v\|_2^2=1$, as happens when each spinor relaxes with spatial separation: in this case $|u(x)|^2$ and $|v(x)|^2$ are the spatial density of each spinorial component, whereas $|u(x)|^2+|v(x)|^2$ gives the combined profile of the two.

It is worth stressing the deliberately `weak' formulation of the preceding interpretation -- which is, at any rate, all what can be said concerning the class of preparable states of interest. If one was to measure the probability that the state $\gamma_N^{(1)}$ is precisely of the form $\begin{pmatrix} u \\ 0 \end{pmatrix}$ or of the form $\begin{pmatrix} 0 \\ v \end{pmatrix}$, this would be given by the numbers
\[
\left\langle \begin{pmatrix} u \\ 0 \end{pmatrix},\,\gamma_N^{(1)}\begin{pmatrix} u \\ 0 \end{pmatrix}\right\rangle_{\!\!L^2(\mathbb{R}^3)\otimes\mathbb{C}^2}=\;\|u\|_2^4\,,\qquad \left\langle \begin{pmatrix} 0 \\ v \end{pmatrix},\,\gamma_N^{(1)}\begin{pmatrix} 0 \\ v \end{pmatrix}\right\rangle_{\!\!L^2(\mathbb{R}^3)\otimes\mathbb{C}^2}=\;\|v\|_2^4\,.
\]

With this analysis in mind, we are now ready to state the mathematical problem for the many-body Schr\"{o}dinger evolution of a state initially prepared with a one-body reduced marginal
\begin{equation}
\gamma_N^{(1)}\;\approx\;\Big|\!\begin{pmatrix} u_0 \\ v_0\end{pmatrix}\!\Big\rangle\Big\langle\!\begin{pmatrix} u_0 \\ v_0\end{pmatrix}\!\Big|\;=\;\begin{pmatrix} |u\rangle\langle u_0| & |u_0\rangle\langle v_0| \\ |v_0\rangle\langle u_0| & |v_0\rangle\langle v_0|\end{pmatrix},\qquad \|u_0\|_2^2+\|v_0\|_2^2=1
\end{equation}
and to formulate our main results. This will be the object of the next Section.

\section{Setting of the problem and main result}\label{sect:main}

We consider a system of $N$ identical spin-$\frac{1}{2}$ bosons in three dimensions. The Hilbert space for this system is
\begin{equation}\label{eq:HspaceHN}
\cH_N\;:=\;\left(L^2(\mathbb{R}^3)\otimes\mathbb{C}^2\right)^{\otimes_{\mathrm{sym}} N}\,,
\end{equation}
as already defined in \eqref{eq:h-one-body}-\eqref{eq:H-many-body}. The system be governed by a Hamiltonian $H$ consisting of a potential part, made of two-body spatial interaction potentials, plus the sum of $N$ one-body Hamiltonians containing a kinetic part, an external spatial trapping potential, and an interaction between the spin of each particle and an external magnetic field. Thus, in self-explanatory notation, and in suitable units,
\begin{equation}\label{eq:unscaled_H}
H\;=\;\sum_{j=1}^N\Big(-\Delta_{x_j}+U^{\mathrm{trap}}(x_j)+\mathbf{B}(x_j,t)\cdot\mathbf{\sigma}_j\Big)+\sum_{j<k}^NV(x_j-x_k)\,.
\end{equation}
Clearly, the part $\sum_{j=1}^N\big(-\Delta_{x_j}\big)+\sum_{j<k}^NV(x_j-x_k)$ only acts non-trivially on the spatial degrees of freedom of an element in $\cH_N$. The external potential be matrix-valued and with the form
\begin{equation}
U^{\mathrm{trap}}(x)\;:=\;\begin{pmatrix} U^{\mathrm{trap}}_\uparrow(x) & 0 \\ 0 & U^{\mathrm{trap}}_\downarrow(x)\end{pmatrix},
\end{equation}
so as to possibly act in a different manner on the spatial parts of each spinor, however inducing no spin flipping. For the $j$-th particle, $\mathbf{\sigma}_j$ denotes the vector $\mathbf{\sigma}=(\sigma^x,\sigma^y,\sigma^z)$ of the Pauli matrices 
\[
\sigma^x=\begin{pmatrix} 0 & 1 \\ 1 & 0\end{pmatrix}, \quad \sigma^y=\begin{pmatrix} 0 & -\ii \\ \ii & 0\end{pmatrix}, \quad \sigma^z=\begin{pmatrix} 1 & 0 \\ 0 & -1\end{pmatrix}
\]
relative to the $j$-th spin degree of freedom, thus acting as the identity on all other spin degrees of freedom. The external magnetic field be the real-valued vector field
\begin{equation}
\mathbf{B}(x,t)\;:=\;(B_1(x,t),B_2(x,t),-V_{\mathrm{hf}}(x,t))
\end{equation}
for suitable functions depending on space and time. The notation is chosen consistently with the experiments -- see \eqref{eq:2ph} above -- where $V_{\mathrm{hf}}(x)\equiv V_{\mathrm{hf}}$ is a uniform field inducing the splitting between the two hyperfine levels, and  $B_1(x,t)-\ii B_2(x,t)\equiv \Omega e^{\ii\omega t}$ is the so-called Rabi field for the spin flipping.  No confusion should occur between the hyper-fine coupling potential $V_{\mathrm{hf}}$ and the pair interaction potential $V$.

The exact control of the dynamics generated by the Hamiltonian $H$ at finite (large) $N$ is clearly out of reach, both analytically and numerically, therefore rigorous conclusions are rather sought in the thermodynamic limit $N\to\infty$. This is part of a long-standing major mathematical problem for the dynamics of a Bose gas \cite{S-2007,S-2008,Benedikter-Porta-Schlein-2015}, and in fact no rigorous control of the thermodynamic limit is known so far. What is mathematically doable and physically still meaningful is to mimic the actual thermodynamic limit with some caricature of it realised by scaling the Hamiltonian with $N$ in such a way to retain at any $N$ an amount of relevant physical features of the system \cite{am_GPlim}.

To this aim, we re-scale as customary the pair interaction potential $V$ in \eqref{eq:unscaled_H} through the so-called `Gross-Pitaevskii scaling' \cite[Chapter 5]{LSeSY-ober}, thus replacing $V$ with the function
\begin{equation}\label{eq:GPscaling}
V_N(x)\;:=\;N^2 V(Nx)\,.
\end{equation}
This produces a realistic model for a Bose gas that is very dilute (the effective range and the scattering length of $V_N$ scales as $N^{-1}$, thus much smaller than the mean inter-particle distance $N^{-1/3}$) and with a strong interaction ($\|V_N\|_{\infty}\sim N^2$). Moreover, by regarding $V_N$ as an approximate delta-distribution with a mean-field pre-factor, namely, $V_N(x)\sim N^{-1}(\int_V)\delta(x)$, one concludes that the contribution of the kinetic part of $H$ (made of $N$ terms) and of the potential part of $H$ ($\sim N^2$ terms) are both of order $\mathcal{O}(N)$, which makes the dynamical problem non-trivial also in the limit. Analogous energetic considerations can be made for the typical ground state energy of a dilute Bose gas, which for fairly general interactions is well-known to be asymptotically given by density $\times$ scattering length, and hence in this scaling is a $\mathcal{O}(N)\times\mathcal{O}(N^{-1})=\mathcal{O}(1)$ quantity.

We therefore re-write the Hamiltonian as the $N$-dependent operator
\begin{equation}\label{eq:scaled_H}
H_N\;:=\;\sum_{j=1}^N(-\Delta_{x_j}+S(x_j,t))+N^2\sum_{j<k}^NV(N(x_j-x_k))
\end{equation}
acting self-adjointly on $\cH_N$, having set for convenience
\begin{equation}\label{eq:matrixS}
S(x,t)\;:=\;\begin{pmatrix} U^{\mathrm{trap}}_\uparrow(x)-V_{\mathrm{hf}}(x,t)& & B_1(x,t)-\ii B_2(x,t) \\\\ B_1(x,t)+\ii B_2(x,t) && U^{\mathrm{trap}}_\downarrow(x)+V_{\mathrm{hf}}(x,t)\end{pmatrix}
\end{equation}
(observe that $S$ coincides formally with its adjoint),
and we consider the Cauchy problem for the associated (linear) Schr\"odinger equation
\begin{equation}\label{eq:Cauchy_problem}
\begin{cases}
\;\;\ii \de_t\Psi_N(t)\;=\;H_N\Psi_N(t)\\
\;\;\Psi_N(0)\;=\;\Psi_{N,0}
\end{cases}
\end{equation}
for a given initial datum $\Psi_{N,0}$. Since $H_N$ may depend on time, suitable conditions on the potential $S(x,t)$ will be assumed so as to ensure that the solution to \eqref{eq:Cauchy_problem} exists and is unique in the strong sense for any time.

Following the discussion of Section \ref{sec:intro}, we are concerned with the class of initial data of the form \eqref{eq:initial0u}, \eqref{eq:psicc-psir}, \eqref{eq:psirr}, or \eqref{eq:psie}, that is, $N$-body states whose associated one-body reduced density matrix $\gamma_{N,0}^{(1)}$, defined as in \eqref{eq:partial_trace1}-\eqref{eq:def_partial_trace_with_basis} above, are rank-one projections, more precisely
\begin{equation}\label{eq:initialgamma}
\lim_{N\to\infty}\gamma_{N,0}^{(1)}\;=\;\Big|\!\begin{pmatrix} u_0 \\ v_0\end{pmatrix}\!\Big\rangle\Big\langle\!\begin{pmatrix} u_0 \\ v_0\end{pmatrix}\!\Big|\,, 
\qquad \|u_0\|_2^2+\|v_0\|_2^2=1\,,
\end{equation}
for given one-body orbitals $u_0$ and $v_0$ in $L^2(\mathbb{R}^3)$. 
Even if a priori the limit in \eqref{eq:initialgamma} can be stated in several inequivalent operator topologies, from the trace norm to the weak operator topology, the bounds
\begin{equation}\label{eq:equivalent-BEC-control}
1-\langle\varphi,\gamma_N^{(1)}\varphi\rangle_{\!L^2(\mathbb{R}^3)\otimes\mathbb{C}^2}\;\leqslant\;\mathrm{Tr}_{\!L^2(\mathbb{R}^3)\otimes\mathbb{C}^2}\big|\,\gamma_{N}^{(1)}-|\varphi\rangle\langle\varphi|\,\big|\;\leqslant\;2\sqrt{1-\langle\varphi,\gamma_N^{(1)}\varphi\rangle_{\!L^2(\mathbb{R}^3)\otimes\mathbb{C}^2}}
\end{equation}
(see, e.g., \cite[Eq.~(1.8)]{M-Olg-2016_2mixtureMF})
show that such a convergence can be monitored equivalently in any of them.

While \eqref{eq:initialgamma} encodes as desired the assumption of complete occupation of the one-body spinor $\begin{pmatrix} u_0 \\ v_0\end{pmatrix}$, it does not select yet the appropriate energy scale for the initial datum, compatibly with the adopted scaling limit. Indeed,  \eqref{eq:initialgamma} also includes completely factorised (uncorrelated) many-body states $\begin{pmatrix} u_0 \\ v_0\end{pmatrix}^{\!\!\!\otimes N}\!\!\in \cH_N$, for which however the scaling \eqref{eq:GPscaling} yields anomalously large asymptotics for the expectation of powers $(H_N/N)^k$ ($k\in\mathbb{N}$) of the energy per particle operator. For example, one finds a linear-in-$N$ energy expectation $\langle H_N\rangle\sim N$ with a constant of proportionality macroscopically different that the expected one (due to the emergence in the limit of the first Born approximation $\frac{1}{8\pi}\int_{\mathrm{R}^3} V\ud x$ of the scattering length $a$ of $V$); analogously, one finds an anomalously large cubic-in-$N$ expectation $\langle H_N^2\rangle\sim N^3$ (it is the potential part in the Hamiltonian \eqref{eq:scaled_H} to give such a contribution).
These behaviours are not typical of the ground state of a Bose gas and are due to the lack of short-scale correlations in the factorised $N$-body state, the presence of which would instead compensate the singular short scale behaviour of $V_N$ as $N\to\infty$. In fact short-scale correlations are shown to form dynamically in a very short transient of time \cite{EMS-2008}: in full analogy to the one-component condensation \cite[Chapter 6]{LSeSY-ober}, it is rather to be expected, and we shall include that in the assumptions on the initial states, that in terms of the many-body energy per particle
\begin{equation}\label{eq:manybody_en_part}
\mathcal{E}_N[\Psi_N]\;:=\;\frac{1}{N}\langle\Psi_N,H_N\Psi_N\rangle\,,\qquad\Psi_N\in\cH_N\,,
\end{equation}
and of the \emph{two-component Gross-Pitaevskii energy functional}
\begin{equation}\label{eq:GPfunctional}
\begin{split}
\mathcal{E}^{\mathrm{GP}}[u,v]\;:=\;&	\int_{\mathbb{R}^3}\!\ud x\,\Big(|\nabla u|^2 +|\nabla v|^2 + 4\pi a\big( |u|^4+2 |u|^2|v|^2+ |v|^4\big)\Big)\\
&\quad +\int_{\mathbb{R}^3}\!\ud x\, \Big\langle\!\!\begin{pmatrix}
 u \\   v
\end{pmatrix}\!,\,S(t)\! \begin{pmatrix}
u\\ v
\end{pmatrix}\!\!\Big\rangle_{\!L^2(\mathbb{R}^3)\otimes\mathbb{C}^2}\qquad u,v\in L^2(\mathbb{R}^3)\,,
\end{split}
\end{equation}
where $a$ is the ($s$-wave) scattering length associated to the potential $V$,
the initial state at $t=0$ satisfy the asymptotics  \eqref{eq:initialgamma} \emph{and}
\begin{equation}\label{eq:asymptotic_energy}
\mathcal{E}_N[\Psi_{N,0}]\;\xrightarrow[]{\;\;\;\;N\to\infty\;\;}\;\mathcal{E}^{\mathrm{GP}}[u_0,v_0]\,.
\end{equation}
This is inspired by the intuition that the ground state energy of $H_N$ is indeed captured asymptotically by the minimum of the functional $\mathcal{E}^{\mathrm{GP}}$, which is a theorem for BEC in one component \cite{LSeSY-ober}, and by the intention to explore the many-body dynamics of initial states that are prepared close to the ground state. For the time being, \eqref{eq:initialgamma} and \eqref{eq:asymptotic_energy} form a working hypothesis that at a high level of confidence is expected to cover precisely the class of initial states of the experiments.

At later times $t>0$ the many-body evolution $\Psi_{N,t}$, i.e., the solution to \eqref{eq:Cauchy_problem} with initial datum $\Psi_{N,0}$, is observed in the experiments to preserve complete BEC, in the sense of one-body marginals, onto a time-dependent one-body spinor $\begin{pmatrix}
u_t\\ v_t
\end{pmatrix}$ whose behaviour is governed by the following system of coupled non-linear cubic Schr\"{o}dinger equations \cite{Ketterle_StamperKurn_SpinorBEC_LesHouches2001, Malomed2008_multicompBECtheory, Hall2008_multicompBEC_experiments,StamperKurn-Ueda_SpinorBose_Gases_2012, pita-stringa-2016}
\begin{equation}\label{eq:GPsystem_extended}
\begin{split}
i\de_t u_t\;&=\;(-\Delta +U^{\mathrm{trap}}_\uparrow)u_t+8\pi a(|u_t|^2+|v_t|^2)u_t-V_{\mathrm{hf}}\,u_t+(B_1-\ii B_2)v_t \\
i\de_t v_t\;&=\;(-\Delta +U^{\mathrm{trap}}_\downarrow) v_t+8\pi a(|u_t|^2+|v_t|^2)v_t+V_{\mathrm{hf}}\,v_t+(B_1+\ii B_2)u_t
\end{split}
\end{equation}
with initial data $u_{t=0}\equiv u_0$ and $v_{t=0}\equiv v_0$.
Here, again, $a$ is the ($s$-wave) scattering length associated to the non-scaled pair potential $V$. The explicit dependence of $U^{\mathrm{trap}}$ and of $\mathbf{B}\equiv(B_1,B_2,-V_{\mathrm{hf}})$ on space and time is omitted for short. We have already commented in Section \ref{sec:intro} that the experiment may well have the trapping potential switched off.
In terms of the matrix-valued potential $S(x,t)$ introduced in \eqref{eq:matrixS} and of the `one-body non-linear Hamiltonians'
\begin{equation}\label{eq:onebody-nonlin-Hamilt}
\begin{split}
h^{(u,v)}_{11}\;&:=\;-\Delta +U^{\mathrm{trap}}_\uparrow+8\pi a(|u_t|^2+|v_t|^2)-V_{\mathrm{hf}}\;=\;-\Delta+S_{11}+8\pi a(|u_t|^2+|v_t|^2) \\
h^{(u,v)}_{22}\;&:=\;-\Delta +U^{\mathrm{trap}}_\downarrow+8\pi a(|u_t|^2+|v_t|^2)+V_{\mathrm{hf}}\;=\;-\Delta+S_{22}+8\pi a(|u_t|^2+|v_t|^2)
\end{split}
\end{equation}
we re-write \eqref{eq:GPsystem_extended} in the compact form
\begin{equation} \label{eq:coupled}
\begin{split}
i\de_t u_t\;&=\;h^{(u,v)}_{11} u_t+ S_{12}v_t \\
i\de_t v_t\;&=\;h^{(u,v)}_{22} v_t+ S_{21}u_t\,.
\end{split}
\end{equation}

In order to prove this picture from first principles, and hence to provide a rigorous derivation of the system \eqref{eq:coupled} as the effective non-linear evolution emerging from the many-body linear Schr\"{o}dinger dynamics \eqref{eq:Cauchy_problem}, one must establish at any time $t>0$ the convergence
\begin{equation}
\gamma^{(1)}_{N,t}\;\xrightarrow[]{\;\;\;\;N\to\infty\;\;}\;\Big|\!\begin{pmatrix} u_t \\ v_t\end{pmatrix}\!\Big\rangle\Big\langle\!\begin{pmatrix} u_t \\ v_t\end{pmatrix}\!\Big|
\end{equation}
of the one-body density matrix $\gamma^{(1)}_{N,t}$ associated with $\Psi_{N,t}$ onto the solution $(u_t,v_t)$ to \eqref{eq:coupled} with initial datum  $(u_0,v_0)$, thus closing the diagram
\begin{equation}\label{scheme_for_marginals}
\begin{CD}
\Psi_N @>\scriptsize\textrm{partial trace}>>\gamma_N^{(1)} @>N\to\infty>> \Big| \!\!\begin{pmatrix} u_0 \\ v_0\end{pmatrix}\!\!\Big\rangle\Big\langle\!\! \begin{pmatrix} u_0 \\ v_0\end{pmatrix}\!\!\Big| \\
@ V\scriptsize\begin{array}{c} \textrm{many-body} \\ \textrm{\textbf{linear} dynamics}  \end{array} VV @V  VV               @VV\scriptsize\begin{array}{c}\textrm{\textbf{non-linear}} \\ \textrm{Schr\"{o}dinger eq.} \end{array}V    \\
\Psi_{\! N,t} @>\scriptsize\textrm{\qquad\qquad\qquad\;}>>\gamma_{\! N,t}^{(1)} @>N\to\infty>> \Big| \!\!\begin{pmatrix} u_t \\ v_t\end{pmatrix}\!\!\Big\rangle\Big\langle\!\! \begin{pmatrix} u_t \\ v_t\end{pmatrix}\!\!\Big|\,.
\end{CD}
\end{equation}

To this aim, we impose the following set of assumptions:
\begin{itemize}
 \item[(A1)] The matrix potential $S\equiv(S_{jk})_{j,k\in\{1,2\}}$ be given with
\begin{equation*}
S_{ij}\in C^1 ( \mathbb{R}_t , L^\infty_x (\mathbb{R}^3))\cap W^{1,\infty}( \mathbb{R}_t , L^\infty_x (\mathbb{R}^3))
\end{equation*}
and $S=S^*$.
 \item[(A2)] The real-valued interaction potential $V$ be given such that $V\in L^\infty(\mathbb{R}^3)$, $V$ has compact support, and for almost every $x\in\mathbb{R}^3$ $V$ is spherically symmetric and $V\geqslant 0$. Let $a$ denote the $s$-wave scattering length associated to $V$ (see, e.g., \cite[Appendix C]{LSeSY-ober}). Correspondingly, let $V_N$ the re-scaled potential \eqref{eq:GPscaling} associated to $V$.
 \item[(A3)] Associated to the potentials fixed in (A1)-(A2), let $H_N$ be the many-body Hamiltonian \eqref{eq:scaled_H} acting at each time $t$ on the $N$-body Hilbert space $\cH_N$ fixed in \eqref{eq:HspaceHN}, let $\mathcal{E}_N$ be the the many-body energy-per-particle functional \eqref{eq:manybody_en_part}, and let $\mathcal{E}^{\mathrm{GP}}$ be the two-component Gross-Pitaevskii energy functional \eqref{eq:GPfunctional}.
 \item[(A4)] Two functions $u_0,v_0\in H^2(\mathbb{R}^3)$ be given with $\|u_0\|^2+\|v_0\|^2=1$ and such that the Cauchy problem associated to the non-linear system \eqref{eq:coupled} with initial datum $(u_0,v_0)$ admits a unique solution $(u,v)$ with $u\equiv u_t(x)$, $v\equiv v_t(x)$, and 
	\begin{equation}\label{eq:well_posedness}
	\begin{pmatrix} u \\ v\end{pmatrix}\in C\big(\mathbb{R}_t,H^2_x(\mathbb{R}^3)\otimes\mathbb{C}^2\big)
\,.
	\end{equation}
 \item[(A5)]  Associated to the spinor $\begin{pmatrix} u_0 \\ v_0\end{pmatrix}$ fixed in (A4), a sequence $(\Psi_{N,0})_{N\in\mathbb{N}}$ of initial $N$-body states be given with $\Psi_N\in\cH_N$ and $\|\Psi_N\|=1$, such that the corresponding sequence $(\gamma_{N,0}^{(1)})_{N\in\mathbb{N}}$ of one-body reduced density matrices satisfies the BEC asymptotics \eqref{eq:initialgamma} in the quantitative form
 	\begin{equation} \label{eq:hypconvergence}
	\mathrm{Tr}\,\Big|\gamma_{N,0}^{(1)}-\Big| \!\!\begin{pmatrix} u_0 \\ v_0\end{pmatrix}\!\!\Big\rangle\Big\langle\!\! \begin{pmatrix} u_0 \\ v_0\end{pmatrix}\!\!\Big|\;\leqslant\;\frac{\;\mathrm{const.}}{N^{\eta_1}}
	\end{equation}
\emph{and} the energy asymptotics \eqref{eq:asymptotic_energy} in the quantitative form
\begin{equation}\label{eq:hypconvergence_energy}
 \big|\,\mathcal{E}_N[\Psi_{N,0}]\to \mathcal{E}^{\mathrm{GP}}[u_0,v_0]\,\big|\;\leqslant\;\frac{\;\mathrm{const.}}{N^{\eta_2}}
\end{equation}
for some constants $\eta_1,\eta_2>0$.
\end{itemize}

Some remarks are in order. First, we have already argued that assumption (A5) is expected to select the class of initial states relevant in the experiments. We underline, in particular, that assumption (A1) includes precisely the experimental potentials $V_{\mathrm{hf}}(x)\equiv V_{\mathrm{hf}}$ and $S_{12}(x,t)=B_1(x,t)-\ii B_2(x,t)\equiv \Omega e^{\ii\omega t}$ for suitable constants $V_{\mathrm{hf}},\Omega,\omega\geqslant 0$. Also the repulsive inter-particle pair interaction is consistent with what observed in the experiments.

As a second important remark, we observe that both the dynamical evolutions we deal with in our assumptions, namely the linear many-body Schr\"{o}dinger dynamics and the non-linear Gross-Pitaevskii dynamics, are well posed. Concerning the former, it can be deduced from (A3) by means of standard arguments (see, e.g., \cite{Aiba-Yajima-2013,Jochen-Griesemer-2014} for a recent discussion) that $H_N$ has a time-independent (dense) domain $\mathcal{D}_N\subset \mathcal{H}_N$ of self-adjointness and there exists a unique unitary propagator for \eqref{eq:Cauchy_problem} on $\cH_N$, that is, a  family $\{U_N(t,s)\,|\,t,s\in\mathbb{R}\}$ of unitaries on $\cH_N$, strongly continuous on $\cH_N$ with respect to $(t,s)$, satisfying $U_N(t,s)U_N(s,r)=U_N(t,r)$ and $U_N(t,t)=\mathbbm{1}$ for any $t,s,r\in\mathbb{R}$, and with the additional properties that, equipping $\mathcal{D}_N$ with the graph norm of $H_N|_{t=0}$, each $U_N(t,s)$ is bounded on $\mathcal{D}$, and for each $\Phi_N\in\mathcal{D}_N$ the function $U_N(t,s)\Phi_N$ is continuous in $\mathcal{D}_N$ with respect to $(t,s)$, it is of class $C^1$ in $\cH_N$, and 
\begin{equation}
\ii\partial_t U_N(t,s)\Phi_N\;=\;H_N U_N(t,s)\Phi_N\,,\qquad \ii\partial_s U_N(t,s)\Phi_N\;=\;- U_N(t,s) H_N \Phi_N\,.
\end{equation}
The non-linear Cauchy problem associated to \eqref{eq:coupled} is well-posed too (in fact, it is defocusing and energy sub-critical), which is seen by exploiting an amount of standard analysis that can be found in the closely related works \cite{Jungel-Weishaupl2013_2compNLS_blowup,Antonelli-Weishaupl-2013,Bunoiu-Precup-2016} and which we do not aim at develop explicitly here. 

Under the above assumptions we are able to prove our main result here below. It is a result of persistence in time of pseudo-spinorial BEC and of rigorous derivation of the non-linear effective dynamics. It is formulated as follows.

\begin{theorem}\label{theorem:main} Consider, for each $N\in\mathbb{N}$, $N\geqslant 2$, the sequence of systems consisting of $N$ spin-$\frac{1}{2}$ identical bosons in three dimensions, subject to the Hamiltonian $H_N$ and initialised at time $t=0$ in the state $\Psi_{N,0}$ of complete BEC onto the one-body spinor $\begin{pmatrix} u_0 \\ v_0\end{pmatrix}$, according to the assumptions (A1)-(A5) above. For each $t>0$ let $\Psi_{N,t}$ be the solution to the many-body Schr\"{o}dinger equation \eqref{eq:Cauchy_problem} with initial datum $\Psi_{N,0}$, let $\gamma_{N,t}^{(1)}$ be the associated one-body reduced density matrix, and let $(u_t,v_t)$ be the solution to the non-linear Gross-Pitaevskii system \eqref{eq:coupled} with initial datum $(u_0,v_0)$. Then, at any $t$,
\begin{equation}\label{eq:thesis}
\lim_{N\to\infty}\gamma_{N,t}^{(1)}\;=\;\Big|\!\begin{pmatrix} u_t \\ v_t\end{pmatrix}\!\Big\rangle\Big\langle\!\begin{pmatrix} u_t \\ v_t\end{pmatrix}\!\Big|
\end{equation}
in trace norm, and 
\begin{equation}\label{eq:thm_thesis}
\lim_{N\to\infty}\mathcal{E}_N[\Psi_{N,t}]\;=\;\mathcal{E}^{\mathrm{GP}}[u_t,v_t]\,.
\end{equation}
\end{theorem}

We shall present the proof of Theorem \ref{theorem:main} in Section \ref{sect:strategy}, after completing a number of preparatory steps in Sections \ref{sec:preliminaries} and \ref{sect:scattering}.

For the time being, let us complete the discussion of our result by highlighting a couple of relevant aspects.

For the technique we use in the proof, an adaptation of the `counting' projection method developed by Pickl \cite{kp-2009-cmp2010,Pickl-JSP-2010,Pickl-LMP-2011,Pickl-RMP-2015}, the precise rate of convergence of the limit \eqref{eq:thesis} remains somewhat implicit: it could be well tracked down through the many inequalities occurring in the proof, but it would turn out to be given by a surely non-optimal inverse power $N^{-\eta}$ for some small $\eta>0$ that depends on $u_0$, $v_0$, and on the potentials chosen in $H_N$. For this reason, even if \eqref{eq:thesis} is quantitative, we omit any reference to the rate of convergence in $N$. For a sharper and more explicit rate it would be of interest to adapt to spinors a different technique for the control of the leading one-body effective dynamics for bosons, which has been developed in the Gross-Pitaevskii scaling in the recent works \cite{B-DO-S-gp-2015,BS-2017}.

Furthermore, Theorem \ref{theorem:main} can be generalised to suitable modifications of the many-body Hamiltonian $H_N$, in particular to the case where particles are charged and hence coupled to the external magnetic field through a minimal coupling, which results in replacing the one-particle kinetic operator $-\Delta$ with the magnetic Laplacian $-\Delta_A:=-(\nabla+\ii A)^2$. For milder scalings in $V_N$ than the Gross-Pitaevskii one this would be a fairly easy application of Pickl's method, and in fact the magnetic Laplacian can be included too with the Gross-Pitaevskii scaling \eqref{eq:GPscaling}, but in this latter case a more careful analysis is needed: to this purpose an amount of previously missing details in the literature have been recently worked out by one of us in \cite{AO-GP_magnetic_lapl-2016volume}.

Last, we remark that when $\mathbf{B}=\mathbf{0}$ in \eqref{eq:unscaled_H} and hence $S_{12}=S_{21}=\mathbb{O}$ in \eqref{eq:matrixS}, no spin flipping is induced any longer by the Hamiltonian and the model becomes the same as that of a mixture of two Bose gas in interaction and no interconversion (i.e., no population change) among species. In this case, as our Theorem \ref{theorem:main} would also give, persistence of BEC in each component emerges as $N\to\infty$, governed by a non-linear system completely analogous to \eqref{eq:GPsystem_extended} but without the Rabi terms: this picture was already proved recently in \cite{M-Olg-2016_2mixtureMF} (and subsequently in \cite{Anap-Hott-Hundertmark-2017}) in the mean-field scaling, and in \cite{AO-GPmixture-2016volume} in the Gross-Pitaevskii scaling.

\section{Preparatory material for the proof}\label{sec:preliminaries}

We begin in this Section the preparation for the proof of our main result, Theorem \ref{theorem:main}, introducing the needed algebraic tools and an amount of technical estimates. This requires an adaptation to the spinor case of the `counting' projection method developed by Pickl  \cite{kp-2009-cmp2010,Pickl-JSP-2010,Pickl-LMP-2011,Pickl-RMP-2015}: this results in a number of straightforward generalisations, plus some non-trivial steps that need be performed in the spinor case and are absent in the scalar case.

We start introducing two key operators that are central in our analysis, namely the time-dependent rank-one orthogonal projection $p_t$ and its complement $q_t$ given by
\begin{equation}
p_t\;:=\;\Big| \!\!\begin{pmatrix} u_t \\ v_t\end{pmatrix}\!\!\Big\rangle\Big\langle\!\! \begin{pmatrix} u_t \\ v_t\end{pmatrix}\!\!\Big|\,,\qquad q_t\;:=\;\mathbbm{1}-p_t\,,
\end{equation}
where $\begin{pmatrix} u_t \\ v_t\end{pmatrix}\in L^2(\mathbb{R}^3)\otimes\mathbb{C}^2$ is the spinor of functions that solve the non-linear Gross-Pitaevskii system \eqref{eq:coupled} with initial datum $(u_0,v_0)$. To make the notation lighter, we shall omit the explicit reference to the time dependence of $p_t$ and $q_t$ on time,and simply write $p$ and $q$.

The operators $p$ and $q$ are naturally lifted onto $\cH_N=\left(L^2(\mathbb{R}^3)\otimes\mathbb{C}^2\right)^{\otimes_{\mathrm{sym}} N}$ in the form of the $2N$ operators
\begin{equation}\label{eq:def_pj_qj}
\begin{array}{l}
p_j:=\underbrace{\mathbbm{1}\otimes\dots\otimes\mathbbm{1}}_{j-1}\otimes \,p \otimes \underbrace{\mathbbm{1}\otimes\dots\otimes\mathbbm{1}}_{N-j} \\
q_j:=\underbrace{\mathbbm{1}\otimes\dots\otimes\mathbbm{1}}_{j-1}\otimes \,q \otimes \underbrace{\mathbbm{1}\otimes\dots\otimes\mathbbm{1}}_{N-j}
\end{array}\quad\qquad j\in\{1,\dots, N\}\,.
\end{equation}
Thus, $p_j$ acts on $\Psi_N\in \cH_N$ as
\[
(p_j\Psi_N)(x_1,\dots,x_N)\;=\;\begin{pmatrix}
u_t(x_j)\\v_t(x_j)
\end{pmatrix}\int_{\mathbb{R}^3} \ud y\:\Big\langle\!\!\begin{pmatrix}
\overline{u_t}(y)\\\overline{v_t}(y)
\end{pmatrix},\Psi_N(x_1,\dots,x_{j-1},y,x_{j+1},\dots,x_N)\Big\rangle_{\!\mathbb{C}_y^{2}}\,.
\]

Next, we introduce the orthogonal projections
\begin{equation}\label{eq:defPk}
P_k\;:=\sum_{a\in\{0,1\}^N \atop \sum_{i=1}^N a_i=k}\:\bigotimes_{i=1}^N\,p_i^{1-a_i}q_i^{a_i}\;=\;(q_1\otimes\cdots\otimes q_k\otimes p_{k+1}\otimes\cdots\otimes p_N)_{\textrm{`sym'}}\qquad k\in\{0,\dots,N\}\,,
\end{equation}
and then set $P_k=\mathbbm{O}$ if $k<0$ or $k>N$. In \eqref{eq:defPk} the symbol `sym' denotes the mere sum (without normalisation factor) of all possible permuted versions of the considered string of $N$ one-body projections.
The Hilbert subspace that $P_k$ projects onto is naturally interpreted as the space of $N$-body states with exactly $k$ particles `out of the condensate', in the sense of orthogonality with respect to the particle $\begin{pmatrix} u_t \\ v_t\end{pmatrix}$. It is also simple to check that 
\begin{equation}\label{eq:properties_of_Pk}
[P_k,P_\ell]\;=\;\delta_{k,\ell}P_k\,,\qquad\sum_{k=0}^N P_k\;=\;\mathbbm{1}\,.
\end{equation}

In the following, whenever no notational confusion arises, we shall omit the tensor product sign $\otimes$ and simply write, for instance, $p_1\cdots p_k q_{k+1}\cdots q_N$ in place of $p_1\otimes\cdots\otimes p_k\otimes q_{k+1}\otimes\cdots\otimes q_N$.

With the $P_k$'s at hand, and fixed a weight function $f:\mathbb{N}\rightarrow\mathbb{R}$, we form the operators
\begin{equation}
\widehat f\;:=\;\sum_{k=0}^N f(k)P_k
\end{equation}
and, for fixed $d\in\mathbbm{N}$, the `shifted' version
\begin{equation}\label{eq:shifted_fd}
\widehat{f_d}\;:=\;\sum_{k=-d}^{N-d} f(k+d)P_k\,.
\end{equation}
Some special choices of the weight $f$ will be useful in the proof. One is
\begin{equation}
n(k)\;:=\;\sqrt{\dfrac{k}{N}}\,,\label{eq:weightn}
\end{equation}
in terms of which one has
\begin{equation}\label{eq:weights}
\widehat{n}^{\,2}\;=\;\frac{1}{N}\sum_{k=0}^Nk\,P_k\;=\;\frac{1}{N}\sum_{k=0}^N\,\sum_{i=1}^N\,q_i\,P_k\;=\;\frac{1}{N}\sum_{i=1}^{N}q_i\,. 
\end{equation}
Thanks to \eqref{eq:weights} and to symmetry, one has 
\begin{equation}\label{eq:relation_q_n}
\langle\Psi_N,q_1\Psi_N\rangle\;=\;\langle\Psi_N,\widehat{n}^{\,2}\Psi_N\rangle\;\leqslant\; \langle\Psi_N,\widehat{n}\,\Psi_N\rangle\,,\qquad \Psi_N\in\cH_N\,.
\end{equation}
Another useful weight is going to be
\begin{equation}\label{eq:weightm}
m(k)\;:=\;\begin{cases}
\sqrt{k/N},\qquad \qquad \qquad \,\,\,\,\,\,\, k\ge N^{1-2\xi}\\\\

\frac{1}{2}(N^{-1+\xi}k+N^{-\xi}),\quad \quad\text{else},
\end{cases}
\end{equation}
for some $\xi>0$ to be chosen sufficiently small. The choice \eqref{eq:weightm} makes the function $\mathbb{R}\ni k\mapsto m(k)$ differentiable.

Let us collect some useful properties of the operators defined above.

\begin{lemma}\label{lemma:tools}
	Let $f,g:\mathbb{N}\rightarrow\mathbb{R}$ be given, together with an operator $A_{ij}$ on $\cH_N$ that acts non-trivially only on the $i$-th and $j$-th particle, for given $i,j\in\{1,\dots,N\}$. One has the following properties.
	\begin{itemize}
		\item[(i)] Commutativity:
		\begin{equation}\label{eq:commutativity_fhat_ghat}
		\widehat{f}\,\widehat{g}\;=\;\widehat{g}\,\widehat{f}\;=\;\widehat{f\,g}\,,
		\end{equation}
\begin{equation}\label{eq:commutativity_fhat_P}
[\widehat{f},p_\ell]\;=\;[\widehat{f},q_\ell]\;=\;[\widehat{f},P_k]\;=\;\mathbbm{O}\qquad\forall \ell\in\{1,\dots,N\}\,,\quad\forall k\in\{0,\dots,N\}\,.
\end{equation}

		\item[(ii)] Shift:
		\begin{equation}
		\widehat{f}\,Q_1\,A_{ij}\,Q_2\;=\;Q_1\,A_{ij}\,Q_2\,\widehat{f}_{z-s}\,, \label{eq:commutation}
		\end{equation}
		where $Q_1,Q_2\in\{p_ip_j,p_iq_j,q_ip_j,q_iq_j\}$, $z$ is the number of $q$'s inside $Q_1$ and $s$ is the number of $q$'s inside $Q_2$.
		\item[(iii)] For $m$ defined in \eqref{eq:weightm}, define
\begin{equation}\label{eq:defn_ma_mb}
\begin{split}
&\widehat{m}^a\;:=\;\widehat{m}-\widehat{m_1}\\
&\widehat{m}^b\;:=\;\widehat{m}-\widehat{m_2}\\
&\widehat{m}^c\;:=\;\widehat{m}-2\widehat{m_2}+\widehat{m_4}\\
&\widehat{m}^d\;:=\;\widehat{m}-\widehat{m_1}-\widehat{m_2}+\widehat{m_3}\\
&\widehat{m}^e\;:=\;\widehat{m}-2\widehat{m_1}+\widehat{m_2}
\end{split}
\end{equation}
		and
		\begin{equation} 
		R_{(ij)}\;:=\;p_ip_j\,\widehat{m}^b+(p_iq_j+q_ip_j)\widehat{m}^a. \label{eq:defnr}
		\end{equation}
		Then one has
		\begin{equation}\label{eq:commut_r}
		[A_{ij},\widehat{m}]\;=\;[A_{ij},R_{(ij)}]\,. 
		\end{equation}
		Moreover, if $K_{hr}$ is a bounded operator on $\cH_N$ acting non-trivially only on the $h$-th and $r$-th particle, with $h\notin\{i,j\}$ and $r\notin\{i,j\}$, one has (see also Remark \ref{rem:R} below)
		\begin{equation}\label{eq:commut_r2}
		\begin{split}
		[K_{hr},R_{(ij)}]\;=\;[K_{hr}\,,\,& p_ip_jp_hp_r\,\widehat{m}^c\\
		+&p_ip_j(p_hq_r+q_hp_r)\,\widehat{m}^d\\
		+&(p_iq_j+q_ip_j)p_hp_r\,\widehat{m}^d\\
		+&(p_iq_j+q_ip_j)(p_hq_r+q_hp_r)\,\widehat{m}^e].
		\end{split}
		\end{equation}
		\item[(iv)] Operator bounds: there exist constants $C,\widetilde{C}>0$ such that
		\begin{align}
			\|\widehat{m}^a\|_{\mathrm{op}}\;&\leqslant\; CN^{-1+\xi}   \label{eq:ma}\\
			\|\widehat{m}^b\|_{\mathrm{op}}\;&\leqslant\; CN^{-1+\xi}  \label{eq:mb}\\
			\|R_{(ij)}\|_{\mathrm{op}}\;&\leqslant\; CN^{-1+\xi} \label{eq:r}\\
			\|\widehat{m}^z\|_{\mathrm{op}}\;&\leqslant\; \widetilde{C}N^{-2+3\xi}\quad\forall z\in\{c,d,e\}   \label{eq:mcde},
		\end{align}
where $\|\;\|_{\mathrm{op}}$ denotes the operator norm.
	\end{itemize}
\end{lemma}

\begin{proof}
Part (i) is an immediate consequence of the mutual orthogonality of the $P_\ell$'s, and of the $p_\ell$'s with the $q_\ell$'s. To establish part (ii) one observes that
\[
P_\ell \,Q_1\,A_{ij}\,Q_2\;=\;Q_1\,A_{ij}\,Q_2\,P_{\ell+s-z}\qquad\forall\ell\in\{0,\dots,N\}\,,
\]
since all the $p_r$'s and the $q_r$'s with $r\notin\{i,j\}$ commute with $A_{ij}$, and the identity above in turn implies the thesis. For part (iii) we compute the difference (recall the notation \eqref{eq:shifted_fd} for the shifted operators)
\[
\begin{split}
[A_{ij},\widehat{m}]-[A_{ij},R_{(ij)}]\;&=\;[A_{ij},\widehat{m}]-[A_{ij},p_ip_j(\widehat{m}-\widehat{m_2})+(p_iq_j+q_ip_j)(\widehat{m}-\widehat{m_1})] \\
&=\;[A_{ij},q_i\,q_j\,\widehat{m}]+[A_{ij},p_i\,p_j\,\widehat{m_2}+(p_i\,q_j+q_i\,p_j)\widehat{m_1}]
\end{split}
\]
and we multiply by $\mathbbm{1}=p_ip_j+p_iq_j+q_ip_j+q_iq_j$ from the left: using the shift property \eqref{eq:commutation} and $p_iq_i=\mathbbm{O}$ we find
\[ 
\begin{split}
p_ip_j&\big( [A_{ij},\widehat{m}]-[A_{ij},R_{(ij)}]\big)\;= \\
&=\;p_ip_jA_{ij}q_iq_j\widehat{m}+p_ip_jA_{ij}p_ip_j\widehat{m}_2-p_ip_j\widehat{m}_2A_{ij}+p_ip_jA_{ij}(p_iq_j+q_ip_j)\widehat{m_1} \\
&=\;p_ip_j\widehat{m}_2A_{ij}q_iq_j+p_ip_j\widehat{m}_2A_{ij}p_ip_j-p_ip_j\widehat{m}_2A_{ij}+p_ip_j\widehat{m}_2A_{ij}(p_iq_j+q_ip_j) \\
&=\;\mathbbm{O}\,,
\end{split}
\]
\[ 
\begin{split}
q_iq_j&\big( [A_{ij},\widehat{m}]-[A_{ij},R_{(ij)}]\big)\;= \\
&=\; q_jA_{ij}q_iq_j\widehat{m}-q_iq_j\widehat{m}\,A_{ij}+q_iq_jA_{ij}p_ip_j\widehat{m_2}+q_iq_jA_{ij}(p_iq_j+q_ip_j)\widehat{m_1} \\
&=\;q_iq_j\widehat{m}A_{ij}q_iq_j-q_iq_j\widehat{m}\,A_{ij}+q_iq_j\widehat{m}A_{ij}p_ip_j+q_iq_j\widehat{m}\,A_{ij}(p_iq_j+q_ip_j)\\
&=\;\mathbbm{O}\,,
\end{split}
\]
and analogously $q_ip_j\big( [A_{ij},\widehat{m}]-[A_{ij},R_{(ij)}]\big)=\mathbbm{O}$. Thus, \eqref{eq:commut_r} follows. To prove \eqref{eq:commut_r2} we use \eqref{eq:defnr} and \eqref{eq:defn_ma_mb} and write
\begin{equation} \label{eq:second_commutator}
\begin{split}
[K_{hr},R_{(ij)}]\;=\;&p_ip_j[K_{hr} ,\,\widehat{m}^b]+(p_iq_j+q_ip_j)[K_{hr},\widehat{m}^a]\\
                    \;=\;&p_ip_j[K_{hr} ,\,\widehat{m}]-p_ip_j[K_{hr} ,\,\widehat{m_2}]+(p_iq_j+q_ip_j)[K_{hr},\widehat{m}]-(p_iq_j+q_ip_j)[K_{hr},\widehat{m_1}].
\end{split}
\end{equation}
Now, we observe that the analogue of \eqref{eq:commut_r} is valid too when $\widehat{m}$ is replaced by a shifted $\widehat{m_k}$, i.e.,
\[
[K_{hr} ,\,\widehat{m_k}]=[K_{hr} ,\,p_hp_r(\widehat{m_k}-\widehat{m_{k+2}})+(p_hq_r+q_hp_r)(\widehat{m_k}-\widehat{m_{k+1}})]
\]
(which is precisely \eqref{eq:commut_r} for $k=0$). This is proven exactly the same way as \eqref{eq:commut_r}. By applying last identity to \eqref{eq:second_commutator}, one gets
\[
\begin{split}
[K_{hr},R_{(ij)}]\;=&\;p_ip_j[K_{hr} ,\,\,p_hp_r(\widehat{m}-\widehat{m_{2}})+(p_hq_r+q_hp_r)(\widehat{m}-\widehat{m_{1}})]\\
&\; -p_ip_j[K_{hr} ,\,\,p_hp_r(\widehat{m_2}-\widehat{m_{4}})+(p_hq_r+q_hp_r)(\widehat{m_2}-\widehat{m_{3}})]\\
&\;+(p_iq_j+q_ip_j)[K_{hr},\,\,p_hp_r(\widehat{m}-\widehat{m_{2}})+(p_hq_r+q_hp_r)(\widehat{m}-\widehat{m_{1}})]\\
&\;-(p_iq_j+q_ip_j)[K_{hr},\,p_hp_r(\widehat{m_1}-\widehat{m_{3}})+(p_hq_r+q_hp_r)(\widehat{m_1}-\widehat{m_{2}})],
\end{split}
\]
which, upon rearrangement, is exactly \eqref{eq:commut_r2}.
As for part (iv), it follows from \eqref{eq:properties_of_Pk} that the $P_\ell$'s produce a direct orthogonal decomposition of $\cH_N$ and hence
\[
\|\widehat{f}\|_{\mathrm{op}}\;=\;\mathop{\text{sup}}_{k\in\{0,\dots, N\}}|f(k)|\,.
\]
As a consequence, treating the function $k\mapsto m(k)$ as continuously defined on $k\in\mathbb{R}$,
\[
\|\widehat{m}^a\|_{\mathrm{op}}\;=\;\sup_{k\in\{0,\dots, N\}}|m(k)-m(k+1)|\;\leqslant\;C\sup_{k\in\{0,\dots, N\}}\big|m'(k)\big|\;\leqslant\; C \,N^{-1+\xi}\,.
\]
A similar reasoning shows that the same bound holds for $\|\widehat{m}^b\|_{\mathrm{op}}$ and hence also for $\|R_{(ij)}\|_{\mathrm{op}}$. Last, we establish \eqref{eq:mcde} starting with $\|\widehat{m}^c\|_{\mathrm{op}}$. By the mean value theorem,
\[
\begin{split}
\|\widehat{m}^c\|_{\mathrm{op}}\;=&\;\sup_{k\in\{0,\dots, N\}}|m(k)-m(k+2)+m(k+4)-m(k+2)|\;=\;\sup_{k\in\{0,\dots, N\}}|-2m'(\beta_k)+2m'(\gamma_k)|\\
\;\leqslant&\;\;\;C\sup_{k\ge N^{1-2\xi}}|m''(\theta_k)|
\end{split}
\]
for some $\beta_k\in(k,k+2)$, $\gamma_k\in(k+2,k+4)$, and $\theta_k\in(k,k+4)$. Since
\[
m''(k)\;=\;\begin{cases}\quad 0\qquad\qquad\qquad\text { for }\,k\;<N^{1-2\xi}\\
\,\,\,-\dfrac{1}{4\sqrt{N\,k^3}}\,\quad\qquad\text { for }k\;>\;N^{1-2\xi},
\end{cases}
\]
$\|\widehat{m}^c\|_{\mathrm{op}}$ is globally bounded by $C\,(Nk^3)^{-1/2}$ when $k=N^{1-2\xi}$, whence \eqref{eq:mcde}. $\|\widehat{m}^d\|_\mathrm{op}$ and $\|\widehat{m}^e\|_\mathrm{op}$ are treated analogously.
\end{proof}

\begin{remark}\label{rem:R}
	The notation for the operator $R_{(ij)}$ is \emph{not} meant to indicate that the only non-trivial action is on the $i$-th and $j$-th variables; it simply indicates that $R_{(ij)}$ depends on $x_i$ and $x_j$ in a more complicated (non-symmetric) way than on all the other variables.
\end{remark}

When dealing with $N$-body wave functions with only partial bosonic symmetry, the following bounds become useful and replace the identity \eqref{eq:weights}.

\begin{lemma} \label{lemma:nq}
	Let $\Psi_N\in L^2(\mathbb{R}^3)^{\otimes N}$ be symmetric with respect to permutations of the $b$ variables $x_{i_1},\dots x_{i_b}$ for some integer $b$ with $2\leqslant b\leqslant N$, and let $f:\mathbb{N}\rightarrow\mathbb{R}$. Then, for any pair of indices $i,j\in\{i_1,\dots,i_b\}$ with $i\neq j$, one has
\begin{eqnarray}
\|\widehat{f}\:q_i\Psi_N\|^2\;&\leqslant\;&\frac{N}{b}\|\widehat{f}\;\widehat{n}\;\Psi_N\|^2\,. \label{eq:nq1} 
\end{eqnarray}
\end{lemma}
\begin{proof}
	The bound  \eqref{eq:nq1} follows from
	\[
	\begin{split}
	\|\widehat{f}\;\widehat{n}\;\Psi_N\|^2\;&=\;\langle\Psi_N,\widehat{n}\;\widehat{f}^{\;2}\;\widehat{n}\;\Psi_N\rangle\;=\;\langle\Psi_N,\widehat{f}^{\;2}\;\widehat{n}^{\,2}\;\Psi_N\rangle\;=\;\frac{1}{N}\sum_{\ell=1}^N\langle\Psi_N,\widehat{f}^{\;2} q_\ell\,\Psi_N\rangle\\
	&\geqslant\;\frac{1}{N}\sum_{k\in\{i_1\dots i_b\}}\langle\Psi_N,\widehat{f}^{\;2}\,q_k\,\Psi_N\rangle\;=\;\frac{b}{N}\,\|\widehat{f}\:q_i\Psi_N\|^2\,,
	\end{split}
	\]
	where we used \eqref{eq:commutativity_fhat_ghat} in the second step, \eqref{eq:weights} in the third, \eqref{eq:commutativity_fhat_P} in the fourth and fifth. 
\end{proof}


Further bounds will turn out to be needed on the operator norm of multiplication operators `dressed' with one or two projections $p$. Such norms are estimated in terms of $\|u_t\|_\infty$ and $\|v_t\|_\infty$, which by assumption (A4) are uniformly bounded quantities in time.

\begin{lemma} \label{lemma:dressed}
	Let $h\in L^1(\mathbb{R}^3)$ and $g\in L^2(\mathbb{R}^3)$, and let $i,j\in\{1,\dots,N\}$. One has
	\begin{align}
		\|g(x_i-x_j)p_j\|_{\mathrm{op}}\;\leqslant\;C(t)\,\|g\|_2\label{eq:dressed2}\\
		\|p_ig(x_i-x_j)\|_{\mathrm{op}}\;\leqslant\;C(t)\,\|g\|_2\label{eq:dressed3}
	\end{align}
for some function $C(t)>0$ depending on $\|u_t\|_\infty$ and $\|v_t\|_\infty$\,, and not on $N$.
\end{lemma}

\begin{proof}
	We observe that
\[
p_j h(x_i-x_j)p_j\;=\;p_j\big(h*|u_t|^2(x_i)+h*|v_t|^2(x_i)\big)\,.
\]
Using this fact and Young's inequality we get
	\[
		\|p_j h(x_i-x_j)p_j\|_{\mathrm{op}}\;\leqslant\;\|p_j\|_{\mathrm{op}}\big\|h*|u_t|^2+h*|v_t|^2\big\|_\infty\;\leqslant\; C\,\|h\|_1(\|u_t\|_\infty^2+\|v_t\|_\infty^2)\,.
	\]
Then \eqref{eq:dressed2} follows from the above inequality (for $h=g^2$) through
	\[
	\begin{split}
	\|g(x_i-x_j)p_j\|_{\mathrm{op}}^2\;&=\;\sup_{\substack{\Psi_N\in\cH_N \\ \|\Psi_N\|=1}}\langle\Psi_N,p_j\,g^2(x_i-x_j)p_j\Psi_N\rangle\;\leqslant\; \|p_jg^2(x_i-x_j)p_j\|_{\mathrm{op}} \\
&\leqslant\; \|g^2\|_1\,(\|u\|_\infty^2+\|v\|_\infty^2)\;\leqslant\; 2\,\|g\|_2^2\,(\|u\|_\infty+\|v\|_\infty)^2\,.
	\end{split}
	\]
	The bound \eqref{eq:dressed3} follows by taking the adjoint of \eqref{eq:dressed2}.
\end{proof}

In particular, the bounds of Lemma \ref{lemma:dressed} above, when applied to the re-scaled potential $V_N$, provide useful $N$-dependent estimates that we collect for convenience here below.

\begin{lemma}\label{lemma:potential} Let $\Psi_N\in\mathcal{D}(H_N)\subset\cH_N$ with $\|\Psi_N\|=1$ and $\mathcal{E}_N[\Psi_N]\leqslant \kappa \,N$ uniformly in $N$ for some $\kappa>0$, and consider the potential $V_N$ defined in \eqref{eq:GPscaling}. Then
\begin{eqnarray}
 \|V_N(x_1-x_2)\Psi_N\|\;&\leqslant&\; C \,N^{1/2}\label{eq:potential} \\
		\|p_1V_N(x_1-x_2)\Psi_N\|\;&\leqslant&\;  C \,N^{-1}\label{eq:dressedpotential}
\end{eqnarray}
for some constant $C>0$ that in \eqref{eq:potential} depends on $\kappa$, on $\|V\|_\infty$\,, and on $\|S\|_{L^\infty_t L^\infty_x}$\,, and in \eqref{eq:dressedpotential} depends additionally on $\mathrm{supp}(V)$ and on the (uniform in time) bound on $\|u_t\|_\infty$ and $\|v_t\|_\infty$\,.
\end{lemma}

\begin{proof}
	To prove \eqref{eq:potential} we combine the estimate
	\[
	\begin{split}
	\|V_N(x_1-x_2)\Psi_N\|^2\;&=\;\|\sqrt{V_N(x_1-x_2)}\sqrt{V_N(x_1-x_2)}\Psi_N\|^2\\
	&\leqslant\; \|\sqrt{V_N(x_1-x_2)}\|_\infty^2\;\|\sqrt{V_N(x_1-x_2)}\;\Psi_N\|^2\\
	&\leqslant \|V\|_\infty \;N^2\,\langle\Psi_N,V_N(x_1-x_2)\Psi_N\rangle 
	\end{split}
	\]
with the estimate
\[
\begin{split}
 \mathcal{E}_N[\Psi_N]\;\geqslant\;-N\|S\|_{L^\infty_t L^\infty_x}+\sum_{i<j}^N\langle\Psi_N,V_N(x_i-x_j)\Psi_N\rangle\,,
\end{split}
\]
thus finding
\[
 \|V_N(x_1-x_2)\Psi_N\|^2\;\leqslant\;\|V\|_\infty \,\big(\mathcal{E}_N[\Psi_N]+N\|S\|_{L^\infty_t L^\infty_x}\big)\;\leqslant\;\|V\|_\infty \,(\kappa+\|S\|_{L^\infty_t L^\infty_x})\:N\,.
\]
To prove \eqref{eq:dressedpotential} we estimate
	\[
	\begin{split}
	\|p_1V_N(x_1-x_2)\Psi_N\|\;\leqslant\;\|p_1\mathbbm{1}_{\mathrm{supp}(V_N)}(x_1-x_2)\|_{\mathrm{op}}\;\|V_N(x_1-x_2)\Psi_N\|\,,
	\end{split}
	\]
where $\mathbbm{1}_{\mathrm{supp}(V_N)}$ is the characteristic function  of the support of $V_N$; the first factor in the r.h.s.~above is estimated, using \eqref{eq:dressed3}, as
\[
\|p_1\mathbbm{1}_{\mathrm{supp}(V_N)}(x_1-x_2)\|_{\mathrm{op}}\;\leqslant\;C\,\|\mathbbm{1}_{\mathrm{supp}(V_N)}\|_2\;\leqslant\;C\,N^{-3/2}
\]
for some constant $C$ depending on (the support of) the non-scaled potential $V$ and on the (uniform in time) bound on $\|u_t\|_\infty$ and $\|v_t\|_\infty$\,, whereas the second factor is estimated as $C\,N^{1/2}$ by \eqref{eq:potential}. The product of these two bounds yields \eqref{eq:dressedpotential}.
\end{proof}


\section{Zero-energy scattering problem and short scale structure}\label{sect:scattering}

As well known (see \cite{LSeSY-ober,EMS-2008,S-2007,S-2008,Benedikter-Porta-Schlein-2015} and the references therein), the understanding of both the ground state and the dynamics of a dilute Bose gas analysed in the Gross-Pitaevskii scaling limit is intimately related to the two-body scattering problem at zero energy. The latter determines the short scale structure of the typical many-body state under consideration, which is crucial to identify the correlation pattern at the leading order in the energy and in the evolution dynamics of the state.

In this Section we collect the main facts from the two-body scattering problem at zero energy needed for the specific technique that we make use of in the present work. To this aim we follow closely the recent works \cite{Pickl-RMP-2015,Jeblick-Leopold-Pickl-2DGP-2016}. 


We start by recalling (e.g., from \cite[Appendix C]{LSeSY-ober}) that given a potential $V$ satisfying our assumption (A2), the scattering length of $V$ is the quantity
\begin{equation}
 a\;:=\;\frac{1}{8\pi}\int_{\mathbb{R}^3} \ud x\, V(x) f(x)\,,
\end{equation}
where $f$ is the so-called zero-energy scattering solution, that is, the solution to the problem
\begin{equation}
\big(-\Delta+{\textstyle\frac{1}{2}}V\big)f\;=\;0\,,\qquad f(x)\xrightarrow[]{\;\;|x|\rightarrow\infty\;\;}1\,.
\end{equation}
By scaling, one sees that the scattering length $a_N$ and the zero-energy scattering solution $f_N$ relative to the re-scaled potential $V_N(x)=N^2 V(Nx)$ are given by
\begin{equation}\label{eq:aN}
a_N\;=\;\frac{a}{N}\,,\qquad\qquad f_N(x)\;=\;f(Nx)\,.
\end{equation}
In particular, $f_N$ has the peculiar structure at the spatial scale $|x|\sim N^{-1}$: in fact,
\begin{equation}\label{eq:fN_2}
f_N(x)\;\underset{|x|\to\infty}{\approx}\;1-\frac{a}{N|x|}\,,\qquad\textrm{and}\qquad 1-\frac{a}{N|x|}\;\leqslant\;f_N(x)\;\leqslant 1\quad\forall x\neq 0\,.
\end{equation}

Along the main proof it is going to be technically convenient to replace the actual potential $V_N$ with a surrogate repulsive potential with a milder scaling and an easier controllability, supported on a spherical shell surrounding, disjointly, the ball of the support of $V_N$. For suitable $\beta\in(0,1)$, let 
\begin{equation} \label{eq:def_W_beta}
W_\beta(x)\;:=\;\begin{cases}
\:4\pi \,a_N \,N^{3\beta} & N^{-\beta}<|x|<R_\beta\\
\;\;0 & \textrm{otherwise}\,.
\end{cases}
\end{equation}
Thus, by construction, for $N$ large enough one has
\begin{equation}
 \text{supp}(V_N)\cap\text{supp}(W_\beta)\;=\;\emptyset\,.
\end{equation}
The spatial scale $R_\beta$ is fixed by the scattering properties of the \emph{difference} potential $V_N-W_\beta$: denoting by $f_\beta$ the zero-energy scattering solution relative to such potential, namely the solution to the problem
\begin{equation}\label{eq:zesp-beta}
\big(-\Delta+{\textstyle\frac{1}{2}}(V_N-W_\beta\big)\big)f_\beta\;=\;0\,,\qquad f_\beta(x)\xrightarrow[]{\;\;|x|\rightarrow\infty\;\;}1\,,
\end{equation}
it can be easily argued that the internal, repulsive potential $V_N$ and the external-shell, attractive potential $-W_\beta$ conspire in \eqref{eq:zesp-beta} so as to yield for $V_N-W_\beta$ a \emph{smaller} scattering length as compared to that of $V_N$; then $R_\beta$ is set precisely to the minimum value above $N^{-\beta}$ which makes the scattering length of $V_N-W_\beta$ \emph{vanish} and hence makes $f_\beta(x)$  \emph{constant} for $|x|>R_\beta$. It can be also proved \cite[Lemma 5.1]{Pickl-RMP-2015} that
\begin{equation} \label{eq:R_beta}
R_\beta\;=\;\mathcal{O}(N^{-\beta})\qquad\textrm{as}\qquad N\to\infty
\end{equation}
and hence the spherical shell where $W_\beta$ is supported in is entirely of the order $N^{-\beta}$.

Crucial when replacing $V_N$ with $W_\beta$ is the function
\begin{equation}
g_\beta\;:=\;1-f_\beta. \label{eq:defng}
\end{equation}
whose relevant properties are collected in the following Lemma.

\begin{lemma}\label{prop:g}
In terms of $a_N$, $f_\beta$, and $g_\beta$ given, respectively, in \eqref{eq:aN}, \eqref{eq:zesp-beta}, and \eqref{eq:defng}, one has
\begin{equation}\label{eq:properties_fb_gb}
0\;\leqslant\;f_N(x)\;\leqslant\;f_\beta(x)\,,\qquad\textrm{whence also}\qquad |g_\beta(x)|\;\leqslant\;\frac{a_N}{|x|}\,\mathbbm{1}_{\{|x|\leqslant R_\beta\}}\qquad\forall x\neq 0\,,
\end{equation}
and
\begin{align}
			\|g_\beta\|_1\;&\leqslant\; C \,N^{-(1+2\beta)} \label{eq:g1}\\
			\|g_\beta\|_{3/2}\;&\leqslant\; C\, N^{-(1+\beta)}\label{eq:g32} \\
			\|g_\beta\|_2\;&\leqslant\; C \,N^{-(1+\frac{1}{2}\beta)}\label{eq:g2}
		\end{align}
for some constant $C>0$ depending on $V$ and $\beta$\,.
\end{lemma}

\begin{proof} The inequalities \eqref{eq:properties_fb_gb} are taken from \cite[Lemma 5.1]{Pickl-RMP-2015}. As for the estimates \eqref{eq:g1}-\eqref{eq:g2}, they follow from \eqref{eq:properties_fb_gb}, for
	%
	\[
	\|g_\beta\|_1\;\leqslant\; a_N\int_0^{R_\beta}\frac{1}{|x|}\,\ud x\;\leqslant\; C\, N^{-1-2\beta}
	\]
and
	\[
	\|g_\beta\|_2^2\;\leqslant\; a^2_N\int_0^{R_\beta}\frac{1}{\,|x|^2}\,\ud x\;\leqslant\;C\, N^{-2-\beta}\,,
	\]
plus $L^1$-$L^2$ interpolation for the $L^{3/2}$-estimate; the constant $C>0$ depends on $V$ and $\beta$.
\end{proof}

\section{Proof of Theorem \ref{theorem:main}}\label{sect:strategy}

We come now to the actual proof of Theorem \ref{theorem:main}. As customary in this scheme, what exactly is going to be controlled is the quantity
\begin{equation} \label{eq:alpha_tilde}
\begin{split}
 \widetilde{\alpha}_N(t)\;&:=\;1-\Big\langle\!\!\begin{pmatrix} u_t \\ v_t\end{pmatrix}\!,\gamma_{N,t}^{(1)}\begin{pmatrix} u_t \\ v_t\end{pmatrix}\!\!\Big\rangle_{\!\!L^2(\mathbb{R}^3\otimes\mathbb{C}^2)} \\
&=\;\frac{1}{N}\sum_{j=1}^N\big\langle\Psi_{N,t},(q_t)_{\!j}\Psi_{N,t}\big\rangle_{\!\cH_N}\;=\;\big\langle\Psi_{N,t},\widehat{n}_t^{\,2}\,\Psi_{N,t}\big\rangle_{\!\cH_N}\,,
\end{split}
\end{equation}
where $\Psi_{N,t}$ is the solution at time $t$ to the many-body Schr\"{o}dinger equation \eqref{eq:Cauchy_problem} with initial datum $\Psi_{N,0}$, $\gamma_{N,t}^{(1)}$ is the associate one-body reduced density matrix, $(u_t,v_t)$ is the solution to the non-linear Gross-Pitaevskii system \eqref{eq:coupled} with initial datum $(u_0,v_0)$, the projections $(q_t)_{\!j}$ with $j\in\{1,\dots,N\}$ are defined in \eqref{eq:def_pj_qj}, and the operator $\widehat{n}_t$ is defined in \eqref{eq:weightn}. It is natural to regard $\widetilde{\alpha}_N(t)$ as the displacement from 100\% of the macroscopic occupation number in the many-body state $\Psi_{N,t}$ of the one-body spinor $\begin{pmatrix} u_t \\ v_t\end{pmatrix}$, thus the vanishing of $\widetilde{\alpha}_N(t)$ has a clear meaning of condensation. In fact, as a consequence of the bounds \eqref{eq:equivalent-BEC-control}, the condition
\begin{equation}\label{eq:convergence}
	\lim_{N\rightarrow\infty}\widetilde{\alpha}_N(t)\;=\;0  
	\end{equation}
and the condition
\begin{equation}\label{eq:thesis_again}
\lim_{N\to\infty}\mathrm{Tr}_{L^2(\mathbb{R}^3\otimes\mathbb{C}^2)}\Big|\gamma_{N,t}^{(1)}-\Big|\!\begin{pmatrix} u_t \\ v_t\end{pmatrix}\!\Big\rangle\Big\langle\!\begin{pmatrix} u_t \\ v_t\end{pmatrix}\!\Big|\;\Big|\;=\;0
\end{equation}
appearing in the actual thesis of Theorem \ref{theorem:main}, are \emph{equivalent}. We shall then prove \eqref{eq:convergence}.

The typical way how the smallness of $\widetilde{\alpha}_N(t)$ is controlled at $t>0$, given its smallness at $t=0$, is a Gr\"onwall-type estimate of the form
\begin{equation}\label{eq:gronwalltilde}
\frac{\ud}{\ud t}\widetilde\alpha_N(t)\;\leqslant\; C\,(\widetilde\alpha_N(t)+N^{-\eta}) 
\end{equation}
for some constants $C,\eta>0$. In the mean-field scaling, when $V_N(x)$ is instead re-scaled as $N^{-1}V(x)$ and the emerging effective dynamics is governed by the non-linear Hartree equation, differentiating in time directly in $\widetilde\alpha_N(t)$ produces terms that are as small as $\widetilde\alpha_N(t)$ itself or as some negative power of $N$, whence the desired estimate \cite{kp-2009-cmp2010,Pickl-LMP-2011}.

However, in the Gross-Pitaevskii scaling, the short-distance behaviour of $V_N$ is so singular in $N$ that a direct differentiation in time yields terms that are not controllable directly by $\widetilde\alpha_N(t)$ and $N^{-\eta}$. Indeed, the scheme of \cite{kp-2009-cmp2010,Pickl-LMP-2011} would rather give, as argued in \cite[Section 6.1.1]{Pickl-RMP-2015}, a bound of the form
\begin{equation}\label{eq:temp_bound}
\frac{\ud}{\ud t}\widetilde\alpha_N(t)\;\leqslant\; C\,\big(\widetilde\alpha_N(t)+\langle\Psi_{N,t},\widehat{n}_t\,\Psi_{N,t}\rangle_{\!\cH_N}+|\mathcal{E}_N[\Psi_{N,t}]-\mathcal{E}^{\mathrm{GP}}[u_t,v_t]\,|+o(1)\big)
\end{equation}
as $N\to\infty$, showing that in order to control the variation of $\widetilde\alpha_N(t)$ one also needs a control on the larger quantity $\langle\Psi_{N,t},\widehat{n}_t\,\Psi_{N,t}\rangle_{\!\cH_N}\geqslant\langle\Psi_{N,t},\widehat{n}_t^{\,2}\,\Psi_{N,t}\rangle_{\!\cH_N}= \widetilde\alpha_N(t)$. In turn, \eqref{eq:temp_bound} suggests that Gr\"onwall Lemma should be rather applied to the quantity
\[
 \langle\Psi_{N,t},\widehat{n}_t\,\Psi_{N,t}\rangle_{\!\cH_N}+|\mathcal{E}_N[\Psi_{N,t}]-\mathcal{E}^{\mathrm{GP}}[u_t,v_t]\,|\,,
\]
except that differentiating it in time would now produce expectations of $\widehat{n}-\widehat{n}_1$ and $\widehat{n}-\widehat{n}_2$ (the shifted operator $\widehat{n}_d$ being defined in \eqref{eq:shifted_fd}), for which the only manageable control would be in terms of the expectation of $N^{-1}\widehat{\partial_k n}$, but the derivative of the weight function $k\mapsto n(k)$ turns out to be too singular at $k=0$ to produce good estimates. Following these considerations, in analogy to the discussion in \cite[Section 6.1.1]{Pickl-RMP-2015}, one is led to select
\begin{equation}\label{eq:def_a<}
\alpha_N^<(t)\;:=\;\langle\Psi_{N,t},\widehat{m}_t\Psi_{N,t}\rangle_{\!\cH_N}+|\mathcal{E}_N[\Psi_{N,t}]-\mathcal{E}^{\mathrm{GP}}[u_t,v_t]\,|
\end{equation}
as a convenient quantity to Gr\"onwall-control in time, where $m(k)$ is the smoothed weight function \eqref{eq:weightm} obtained by regularising $n(k)$ at small $k$. Let us recall that by construction
\begin{equation}\label{eq:m-n-n2}
\max\{n(k),N^{-\xi}\}\;\geqslant\; m(k)\;\geqslant\;n(k)\;\geqslant\;n^2(k)\,,\qquad k\in[0,N]\,.
\end{equation}

The choice of the control \eqref{eq:def_a<} with the weight \eqref{eq:weightm} turns out to bring an efficient Gr\"onwall-like bound to prove the limits \eqref{eq:thesis}-\eqref{eq:thm_thesis} of Theorem \ref{theorem:main} \emph{provided that} the potential $V_N(x)=N^2V(N x)$ is replaced with a softer scaling potential 
\begin{equation}\label{eq:def_Vtilde}
 \widetilde{V}_N(x)\;:=\;N^{3\gamma-1}V(N^\gamma x)
\end{equation}
defined for some $\gamma\in(0,1)$, as can be seen by reasoning as in \cite[Sections 6.1.1 and 6.1.2]{Pickl-RMP-2015}. However, in the actual Gross-Pitaevskii scaling (i.e., $\gamma=1$ in \eqref{eq:def_Vtilde}) a further modification of $\alpha_N^<(t)$ is necessary, for otherwise the peculiar short scale structure induced by the zero-energy scattering problem associated with $V_N$ -- see \eqref{eq:aN}-\eqref{eq:fN_2} above -- would prevent us to close a Gr\"onwall-type argument for $\alpha_N^<(t)$.

In \cite{Pickl-RMP-2015} this difficulty for the one-component condensate is cleverly circumvented by `dressing' the projections $p_j$ (that `count' the particles in the condensate) with a typical Jastrow factor built upon the zero-energy scattering solution $f_\beta$ defined in \eqref{eq:zesp-beta}, relative to the smoothed potential $V_N-W_\beta$ defined in \eqref{eq:def_W_beta}. In analogy to that, also in the spinor case we shall make the replacement
\[
p_1\;\mapsto\;\prod_{j=2}^N\,f_\beta(x_1-x_j)p_1\prod_{k=2}^Nf_\beta(x_1-x_k)\,,
\]
the precise value of $\beta$ to be fixed conveniently. Let us recall from Section \ref{sect:scattering} that $f_\beta$ is actually constant in the outer region $|x|>\mathcal{O}(N^{-\beta})$ and has a smoothed behaviour as $|x|\to 0$.

If now one was to re-do the projection-based analysis of \cite{kp-2009-cmp2010,Pickl-LMP-2011} with the insertion in $\frac{\ud}{\ud t}\widetilde{\alpha}_N(t)$ of such dressed projections, then, as shown in \cite[Section 6.2.1]{Pickl-RMP-2015}, one would get terms of the form
\[
\langle\Psi_{N,t},(q_t)_1\Psi_{N,t}\rangle_{\!\cH_N}+2(N-1)\,\mathfrak{Re}\langle\Psi_{N,t},g_\beta(x_1-x_2)(p_t)_1\Psi_{N,t}\rangle_{\!\cH_N}
\]
up to three-body re-collision terms. This finally motivates the following:

\begin{definition} For $\beta\in(0,1)$ and $\Psi_{N,t}$ as in the assumptions of Theorem \ref{theorem:main}, we define at each time $t$
	\begin{equation}\label{eq:def_alphaN}
\begin{split}
 \alpha_N(t)\;&:=\;\alpha^<_N(t)-N(N-1)\,\mathfrak{Re}\langle\Psi_{N,t},g_\beta(x_1-x_2)\,R_{(12),t}\,\Psi_{N,t}\rangle_{\!\cH_N} \\
&=\;\langle\Psi_{N,t},\widehat{m}_t\Psi_{N,t}\rangle_{\!\cH_N}+|\mathcal{E}_N[\Psi_{N,t}]-\mathcal{E}^{\mathrm{GP}}[u_t,v_t]\,| \\
&\qquad\quad -N(N-1)\,\mathfrak{Re}\langle\Psi_{N,t},g_\beta(x_1-x_2)\,R_{(12),t}\,\Psi_{N,t}\rangle_{\!\cH_N} \,,
\end{split}
	\end{equation}
	where $\alpha^<_N(t)$ is defined in \eqref{eq:def_a<}, $m(k)$ is the weight function \eqref{eq:weightm}, $g_\beta$ is the cut-off function defined in \eqref{eq:defng}, and $R_{(ij),t}$ is the operator in \eqref{eq:defnr}.
\end{definition}

We have thus \emph{three} indicators, $\widetilde{\alpha}_N(t)$, $\alpha^<_N(t)$, and $\alpha_N(t)$.
First, we see that $\alpha_N(t)$ and $\alpha^<_N(t)$ are close and coincide asymptotically as $N\to\infty$.

\begin{lemma}\label{lem:aa<}
 Under the assumptions of Theorem \ref{theorem:main}, and for any $\beta\in(0,1)$ chosen in the definition \eqref{eq:def_alphaN} of $\alpha_N(t)$, there exist a constant $\eta>0$ and a function $C(t)>0$, that depends only on $\|u_t\|_\infty$\,, $\|v_t\|_\infty$\,, $V$, and $\beta$ and is independent of $N$, such that, for $N$ large enough, one has
\begin{equation} \label{eq:apriori}
|\alpha_N(t)-\alpha^<_N(t)|\;\leqslant\;C(t)\,N^{-\eta}\,.
\end{equation}
\end{lemma}

\begin{proof} From \eqref{eq:def_alphaN} one has
\begin{equation*}
 \begin{split}
  |\alpha_N(t)-\alpha^<_N(t)|\;&\leqslant\;N^2\,\big| \langle\Psi_{N,t},g_\beta(x_1-x_2)\,R_{(12),t}\,\Psi_{N,t}\rangle_{\!\cH_N}  \big|\,, 
 \end{split}
\end{equation*}
and from \eqref{eq:defnr} one has
\[
 \begin{split}
  \big| \langle\Psi_{N,t},g_\beta(x_1-x_2)\,R_{(12),t}\,\Psi_{N,t}\rangle_{\!\cH_N}  \big| \;&\leqslant\;\big| \langle\Psi_{N,t},g_\beta(x_1-x_2)\,(p_t)_1(p_t)_2\,\widehat{m}^b_t\,\Psi_{N,t}\rangle_{\!\cH_N}  \big| \\
&\quad + \big| \langle\Psi_{N,t},g_\beta(x_1-x_2)\,(p_t)_1(q_t)_2\,\widehat{m}^a_t\,\Psi_{N,t}\rangle_{\!\cH_N}\big|   \\
&\quad + \big| \langle\Psi_{N,t},g_\beta(x_1-x_2)\,(q_t)_1(p_t)_2\,\widehat{m}^a_t\,\Psi_{N,t}\rangle_{\!\cH_N} \big| \\
&\leqslant\;\|g_\beta(x_1-x_2)\,(p_t)_1\|_{\mathrm{op}}\:\big(\|\widehat{m}^b_t\|_{\mathrm{op}}+\|\widehat{m}^a_t\|_{\mathrm{op}}\big) \\
&\qquad + \|(p_t)_2\: g_\beta(x_1-x_2)\|_{\mathrm{op}}\:\|\widehat{m}^a_t\|_{\mathrm{op}}\,.
 \end{split}
\]
Therefore,
\[
 \begin{split}
  |\alpha_N(t)-\alpha^<_N(t)|\;&\leqslant\;N^2\big(\|g_\beta(x_1-x_2)\,(p_t)_1\|_{\mathrm{op}}+ \|(p_t)_2\: g_\beta(x_1-x_2)\|_{\mathrm{op}}\big)\:\big(\|\widehat{m}^b_t\|_{\mathrm{op}}+\|\widehat{m}^a_t\|_{\mathrm{op}}\big) \\
&\leqslant\;N^2\,C(t)\,\|g_\beta\|_2\,N^{-1+\xi}\;\leqslant\;N^2\,C(t)\,N^{-1-\frac{1}{2}\beta}\,N^{-1+\xi}\;=\;C(t)\,N^{-\frac{1}{2}\beta+\xi}\,,
 \end{split}
\]
where we used \eqref{eq:ma}-\eqref{eq:mb} and \eqref{eq:dressed2}-\eqref{eq:dressed3} in the second inequality, and \eqref{eq:g2} in the third inequality, for some function $C(t)>0$ depending only on $\|u_t\|_\infty$, $\|v_t\|_\infty$, $V$, and $\beta$. Here $\xi$ is the constant used in the definition \eqref{eq:weightm} of the weight $m(k)$. By taking it small enough, i.e., $0<\xi<\frac{1}{2}\beta$, one obtains the constant $\eta:=\frac{1}{2}\beta-\xi>0$ of the thesis.
\end{proof}

Next, we can prove the following estimate.

\begin{proposition} \label{prop:bound}
Under the assumptions of Theorem \ref{theorem:main},  there exists $\beta\in(0,1)$ such that the bound
\begin{equation}\label{eq:gronwall}
 |\alpha_N(t)|\;\leqslant\;C(t)\big(\alpha_N^<(0)+|\mathcal{E}_N[\Psi_{N,0}]-\mathcal{E}^{\mathrm{GP}}[u_0,v_0]|+N^{-\eta}\big)+\int_0^tC(s)\,\alpha_N^<(s)\,\ud s
\end{equation}
holds for some constants $\eta>0$ and $N_0\in\mathbb{N}$, and some function $C(t)>0$ depending on $\|u_t\|_{H^2}$\,, $\|v_t\|_{H^2}$\,, and $V$, 
but not on $N$.
\end{proposition}

The proof of Proposition \ref{prop:bound} is the subject of Section \ref{sect:estimate}. 
With Lemma \ref{lem:aa<} and Proposition \ref{prop:bound}  at hand, we are now ready to prove Theorem \ref{theorem:main}.

\begin{proof}[Proof of Theorem \ref{theorem:main}]
By means of the comparison \eqref{eq:apriori}, the bound \eqref{eq:gronwall} can be turned into
\[\tag{*}
 \alpha_N^<(t)\;\leqslant\;\widehat{C}(t)\big(\alpha_N^<(0)+|\mathcal{E}_N[\Psi_{N,0}]-\mathcal{E}^{\mathrm{GP}}[u_0,v_0]|+N^{-\eta'}\big)+\int_0^t\widehat{C}(s)\,\alpha_N^<(s)\,\ud s
\]
for suitable $\widehat{C}(t)>0$ and $\eta'>0$.
The assumption (A5) now guarantees that the terms in the r.h.s.~of (*) above which are evaluated at $t=0$ are small in $N$. This is clear for the energy difference, because of the bound \eqref{eq:hypconvergence_energy}, whereas concerning $\alpha_N^<(0)$ we argue as follows. First we estimate
\[
\begin{split}
 \langle\Psi_{N,0},\widehat{m}\:\Psi_{N,0}\rangle_{\!\cH_N}\;&\leqslant\;N^{-\xi}+\langle\Psi_{N,0},\widehat{n}\:\Psi_{N,0}\rangle_{\!\cH_N}\;\leqslant\;N^{-\xi}+\langle\Psi_{N,0},\widehat{n}^{\,2}\:\Psi_{N,0}\rangle_{\!\cH_N}^{\!1/2} \\
&=\;N^{-\xi}+\sqrt{\widetilde{\alpha}_N(0)}\;\leqslant\;N^{-\xi}+\sqrt{\mathrm{Tr}_{L^2(\mathbb{R}^3)\otimes\mathbb{C}^2}\Big|\gamma_{N,0}^{(1)}-\Big| \!\!\begin{pmatrix} u_0 \\ v_0\end{pmatrix}\!\!\Big\rangle\Big\langle\!\! \begin{pmatrix} u_0 \\ v_0\end{pmatrix}\!\!\Big|}\;,
\end{split}
\]
where we used \eqref{eq:m-n-n2} in the first inequality, then a Schwarz inequality, then \eqref{eq:alpha_tilde} in the third step, and finally \eqref{eq:equivalent-BEC-control} in the last inequality, and where $\xi>0$ is the constant used in the definition \eqref{eq:weightm} of the weight $m(k)$). Then by \eqref{eq:hypconvergence} and \eqref{eq:hypconvergence_energy} of assumption (A5), 
\[
\begin{split}
 \alpha^<_N(0)\;&=\;\langle\Psi_{N,0},\widehat{m}\:\Psi_{N,0}\rangle_{\!\cH_N}+|\mathcal{E}_N[\Psi_{N,0}]-\mathcal{E}^{\mathrm{GP}}[u_0,v_0]\,| \\
&\leqslant\;C \,(N^{-\xi}+N^{-\frac{1}{2}\eta_1}+N^{-\eta_2})\;\leqslant\; C\,N^{-\eta''}
\end{split}
\]
for some $\eta''>0$.
With these bounds at $t=0$ the inequality (*) takes the form
\[
 \alpha_N^<(t)\;\leqslant\; \widetilde{C}(t)\,N^{-\eta}+\int_0^t\widetilde{C}(s)|\alpha_N^<(s)|\,\ud s
\]
for suitable $\widetilde{C}(t)>0$ and $\eta>0$.
A Gr\"onwall-like estimate \cite[Theorem 1.3.2]{Pachpatte-ineq} then gives 
\[
 \alpha_N^<(t)\;\leqslant\;  \widetilde{C}(t)\,N^{-\eta}+\int_0^t\ud s\:\widetilde{C}(s)^2\,N^{-\eta}e^{\int_s^t\widetilde{C}(r)\ud r}\;\equiv\; C(t)\,N^{-\eta}\,,
\]
having set $C(t)$ accordingly. 
Owing to \eqref{eq:def_a<}, the latter estimate implies at once both
\[
 C(t)\,N^{-\eta}\;\geqslant\;|\mathcal{E}_N[\Psi_{N,t}]-\mathcal{E}^{\mathrm{GP}}[u_t,v_t]\,|\,,
\]
which yields the limit \eqref{eq:thm_thesis} of the statement of Theorem \ref{theorem:main}, and
\[
\begin{split}
 C(t)\,N^{-\eta}\;&\geqslant\;\langle\Psi_{N,t},\widehat{m}_t\Psi_{N,t}\rangle_{\!\cH_N}\;\geqslant\;\langle\Psi_{N,t},\widehat{n}_t^2\Psi_{N,t}\rangle_{\!\cH_N}\;=\;\widetilde{\alpha}_N(t) \\
&\geqslant\;\Big(\mathrm{Tr}_{L^2(\mathbb{R}^3)\otimes\mathbb{C}^2}\Big|\gamma_{N,t}^{(1)}-\Big| \!\!\begin{pmatrix} u_t \\ v_t\end{pmatrix}\!\!\Big\rangle\Big\langle\!\! \begin{pmatrix} u_t \\ v_t\end{pmatrix}\!\!\Big|\,\Big)^{\!2}
\end{split}
\]
(where we used \eqref{eq:m-n-n2} in the second step, \eqref{eq:alpha_tilde} in the third, and \eqref{eq:equivalent-BEC-control} in the fourth), which yields the limit \eqref{eq:convergence} and hence also the limit \eqref{eq:thesis} of the statement of Theorem \ref{theorem:main}.
\end{proof}

\section{Proof of Proposition \ref{prop:bound}}\label{sect:estimate}

This Section is devoted to the proof of Proposition \ref{prop:bound}, which is the missing step to complete the proof of Theorem \ref{theorem:main}.


In order to produce the estimate \eqref{eq:gronwall}, it is convenient to re-express the quantity $\alpha_N(t)$ defined in \eqref{eq:def_alphaN} as
\begin{equation}\label{eq:alpha_beta_energy}
 \alpha_N(t)\;=\;|\mathcal{E}_N[\Psi_{N,t}]-\mathcal{E}^{\mathrm{GP}}[u_t,v_t]\,|+\delta_N(t)\,,
\end{equation}
where
\begin{equation}\label{eq:def_deltaN}
\delta_N(t)\;:=\; \langle\Psi_{N,t},\widehat{m}_t\Psi_{N,t}\rangle_{\!\cH_N}-N(N-1)\,\mathfrak{Re}\langle\Psi_{N,t},g_\beta(x_1-x_2)\,R_{(12),t}\,\Psi_{N,t}\rangle_{\!\cH_N} \,,
\end{equation}
and to analyse the two summands in the r.h.s.~of \eqref{eq:alpha_beta_energy} separately.

In analogy to \cite{Pickl-RMP-2015}, we introduce the following $(N,t)$-dependent quantities.


%

\begin{itemize}
	\item[(a)] A quantity that, as shown later, controls the energy difference in \eqref{eq:alpha_beta_energy}, and precisely

 \begin{equation}\label{eq:def_dN_a}
  \delta^{(a)}_{N}\!(t)\;:=\;\langle\Psi_{N,t}, \dot S(x_1,t)\Psi_{N,t}\rangle_{\!\cH_N}-\Big\langle\!\! \begin{pmatrix}u_t\\v_t\end{pmatrix}\! ,\,\dot S(t)\begin{pmatrix}u_t\\v_t\end{pmatrix}\!\!\Big\rangle_{\!L^2(\mathbb{R}^3)\otimes\mathbb{C}^2}\,.
 \end{equation}
Here, consistently with the notation of Sections \ref{sec:intro} and \ref{sect:main},  $S(x_j,t)$ denotes the operator-valued matrix $S$ (defined in \eqref{eq:matrixS}) acting on the spatial+spin degrees of freedom of the $j$-th particle.
	\item[(b)] A `modified interaction term', containing the new potential $W_\beta$ defined in \eqref{eq:def_W_beta} as well as the function $f_\beta$ introduced in \eqref{eq:zesp-beta}:
	\begin{equation}\label{eq:def_dN_b}
	\begin{split}
\delta_N^{(b)}\!(t)\;:=&\;-N(N-1)\,\mathfrak{Im}\Big(\langle\Psi_{N,t},Z^{(\beta)}_{N,t}(x_1,x_2)R_{(12),t}\Psi_{N,t}\rangle_{\!\cH_N}+ \\
&\qquad\qquad\qquad\qquad+ \langle\Psi_{N,t},g_\beta(x_1-x_2)R_{(12),t}\,Z_{N,t}(x_1,x_2)\Psi_{N,t}\rangle_{\!\cH_N}\Big)\,,
\end{split}
\end{equation}
	where
	\begin{equation}
	\begin{split}\label{eq:defz_beta}
	Z^{(\beta)}_{N,t}&(x_1,x_2)\;:=\;f_\beta(x_1-x_2)\times  \\
& \; \times \Big(W_\beta(x_1-x_2)-{\textstyle\frac{8\pi a}{N-1}} (|u_t(x_1)|^2+|v_t(x_1)|^2+|u_t(x_2)|^2+|v_t(x_2)|^2)\Big)
	\end{split}
	\end{equation}
	and
	\begin{equation} \label{eq:defz}
	Z_{N,t}(x_1,x_2)\;:=\;V_N(x_1-x_2)-{\textstyle\frac{8\pi a}{N-1}}\Big(|u_t(x_1)|^2+|v_t(x_1)|^2+|u_t(x_2)|^2+|v_t(x_2)|^2\Big)\,.
	\end{equation}
	\item[(c)] A term containing mixed spatial derivatives of $g_\beta$ and $R_{(12),t}$\,:
	\begin{equation}\label{eq:def_dN_c}
\delta_N^{(c)}\!(t)\;:=\;-4N(N-1)\mathfrak{Im}\langle\Psi_{N,t},\nabla_1g_\beta(x_1-x_2)\nabla_1R_{(12),t}\Psi_{N,t}\rangle_{\!\cH_N}\,.
\end{equation}
	\item[(d)] A `three particle correction' term:
	\begin{equation}\label{eq:def_dN_d}
	\begin{split}
	\delta_N^{(d)}\!(t)\;:=\;N(N-1)&(N-2)\, \mathfrak{Im}\,\langle\Psi_{N,t},g_\beta(x_1-x_2)\\
	&\times[V_N(x_1-x_3)+8\pi a(|u_t(x_3)|^2+|v_t(x_3)|^2),R_{(12),t} ]\Psi_{N,t}\rangle_{\!\cH_N}\,.
	\end{split}
	\end{equation}
	\item[(e)] A `four particle correction' term:
	\begin{equation}\label{eq:def_dN_e}
	\begin{split}
	\delta_N^{(e)}\!(t)\;&:=\;{\textstyle\frac{1}{2}}\,N(N-1)(N-2)(N-3)\,\times \\
&\qquad\times\,\mathfrak{Im}\,\langle\Psi_{N,t},g_\beta(x_1-x_2)[V_N(x_3-x_4),R_{(12),t} ]\Psi_{N,t}\rangle_{\!\cH_N}\,.
	\end{split}
	\end{equation}
	\item[(f)] A `correction term' for the mean-field potential:
	\begin{equation}\label{eq:def_dN_f}
\begin{split}
	&\!\!\!\!\!\!\!\!\!\!\delta_N^{(f)}\!(t)\;:=\;N(N-2)\,\times \\
&\!\!\!\!\!\!\!\!\!\!\times\,\mathfrak{Im}\langle\Psi_{N,t},g_\beta(x_1-x_2)\big(|u_t(x_1)|^2+|u_t(x_2)|^2+|v_t(x_1)|^2+|v_t(x_2)|^2,R_{(12),t} \big)\Psi_{N,t}\rangle_{\!\cH_N}\,.
\end{split}
	\end{equation}
\end{itemize}

Concerning the quantities above, we shall establish the following three results.

\begin{lemma}\label{lem:enest}
 Under the assumptions of Theorem 2.1, one has
\begin{equation}\label{eq:enest}
 \frac{\ud}{\ud t}\big(\mathcal{E}_N[\Psi_{N,t}]-\mathcal{E}^{\mathrm{GP}}[u_t,v_t]\big)\;=\;\delta^{(a)}_{N}\!(t)\,,\qquad t\geqslant 0\,.
\end{equation}
%
\end{lemma}

%

\begin{proposition} \label{prop:alphaestimate} 
Under the hypothesis of Theorem \ref{theorem:main}, one has
	\begin{equation}\label{eq:ddelta_abcdef}
	\frac{\ud}{\ud t}\delta_N(t)\;=\; \delta_N^{(b)}\!(t)+\delta_N^{(c)}\!(t)+\delta_N^{(d)}\!(t)+\delta_N^{(e)}\!(t)+\delta_N^{(f)}\!(t)\,,\qquad t\geqslant 0\,.
	\end{equation}
\end{proposition}

\begin{proposition}\label{prop:each_g_abcdef}
 Under the hypothesis of Theorem \ref{theorem:main}, for any $j\in\{a,b,c,d,e,f\}$ one has
\begin{equation}\label{eq:bound_dN}
 |\delta_N^{(j)}\!(t)|\;\leqslant\;C(t)(\alpha^<_N(t)+N^{-\eta})\,,\qquad t\geqslant 0,\quad N\geqslant N_0\,,
\end{equation}
for some constants $\eta>0$ and $N_0\in\mathbb{N}$, and some function $C(t)>0$ depending on $\|u_t\|_{H^2}$\,, $\|v_t\|_{H^2}$\,, and $V$, but not on $N$.
\end{proposition}

With Lemma \ref{lem:enest}, Proposition \ref{prop:alphaestimate}, and Proposition \ref{prop:each_g_abcdef} at hand, we are able to obtain \eqref{eq:gronwall}.

\begin{proof}[Proof of Proposition \ref{prop:bound}]
Integrating \eqref{eq:ddelta_abcdef} in time and using the bounds \eqref{eq:bound_dN} yields
\[
 |\delta_N(t)|\;\leqslant \;C(0)(\alpha^<_N(0)+N^{-\eta})+\int_0^t C(s)(\alpha^<_N(s)+N^{-\eta})\,\ud s\,.
\]
In turn, integrating  \eqref{eq:enest} in time and using the bound \eqref{eq:bound_dN} 
yields
\[
 |\mathcal{E}_N[\Psi_{N,t}]-\mathcal{E}^{\mathrm{GP}}[u_t,v_t]\,|\;\leqslant\;|\mathcal{E}_N[\Psi_{N,0}]-\mathcal{E}^{\mathrm{GP}}[u_0,v_0]\,|+\int_0^t C(s)(\alpha^<_N(s)+N^{-\eta})\,\ud s\,.
\]
Combining the last two inequalities above, and using \eqref{eq:alpha_beta_energy}, one then obtains
\[
 |\alpha_N(t)|\;\leqslant\;\widetilde{C}(t)\big(\alpha_N^<(0)+|\mathcal{E}_N[\Psi_{N,0}]-\mathcal{E}^{\mathrm{GP}}[u_0,v_0]|+N^{-\eta}\big)+\int_0^t\widetilde{C}(s)\,\alpha_N^<(s)\,\ud s
\]
for suitable $\widetilde{C}(t)>0$, thus concluding the proof.
\end{proof}

To complete our programme, we pass now to the proof of Lemma \ref{lem:enest}, Proposition \ref{prop:alphaestimate}, and Proposition \ref{prop:each_g_abcdef}.

\begin{proof}[Proof of  Lemma \ref{lem:enest}]
Let us consider the time derivative of each energy functional separately. For  $\mathcal{E}_N[\Psi_{N,t}]=N^{-1}\langle\Psi_{N,t},H_N\Psi_{N,t}\rangle_{\!\cH_N}$, the action of the time derivative on the two vectors $\Psi_{N,t}$ produces a null term, due to \eqref{eq:Cauchy_problem} and to the self-adjointness of $H_N$, so what remains is the expectation of the time derivative of the time-dependent part of $H_N$ itself. Owing to the bosonic symmetry, this is precisely
 \begin{equation}\label{eq:derivative_linear}
\frac{\ud}{\ud t}\mathcal{E}_N[\Psi_{N,t}]\;=\;\langle\Psi_{N,t}, \dot S(x_1,t)\Psi_{N,t}\rangle_{\!\cH_N}\,.
 \end{equation}
As for $ \mathcal{E}^{\mathrm{GP}}$, let us introduce the spinorial Hamiltonian
\begin{equation} \label{eq:spin_hamiltonian}
h^{\mathrm{GP}}\;:=\;\begin{pmatrix}
h^{(u,v)}_{11}&S_{12}\\S_{21}&h^{(u,v)}_{22}
\end{pmatrix}
\end{equation}
with entries defined in \eqref{eq:matrixS} and \eqref{eq:onebody-nonlin-Hamilt}. At each $t\geqslant 0$ the operator $h^{\mathrm{GP}}(t)$ acts on the one-body Hilbert space $L^2(\mathbb{R}^3)\otimes\mathbb{C}^2$. In terms of $h^{\mathrm{GP}}$, 
 \[
 \mathcal{E}^{\mathrm{GP}}[u_t,v_t]\;:=\;\Big\langle\!\!\begin{pmatrix}
 u_t \\   v_t
 \end{pmatrix}\!,\bigg[\,(h^{\mathrm{GP}})(t)-\begin{pmatrix}
 4\pi a\big( |u_t|^2+|v_t|^2\big) &0\\0&4\pi a\big( |u_t|^2+|v_t|^2\big)
 \end{pmatrix}\bigg]\! \begin{pmatrix}
 u_t\\ v_t
 \end{pmatrix}\!\!\Big\rangle_{\!L^2(\mathbb{R}^3)\otimes\mathbb{C}^2}\,.
 \]
Then, using \eqref{eq:coupled},
 \begin{equation}\label{eq:time_derivative_partial}
 \begin{split}
 &\frac{\ud}{\ud t}\mathcal{E}^{\mathrm{GP}}[u_t,v_t]\;=\; \\
&\quad=\; -\ii\,\Big\langle\!\!\begin{pmatrix}
 u_t \\   v_t
 \end{pmatrix}\!,\bigg[\,h^{\mathrm{GP}}-\begin{pmatrix}
 4\pi a\big( |u_t|^2+|v_t|^2\big) &0\\0&4\pi a\big( |u_t|^2+|v_t|^2\big)
 \end{pmatrix}, h^{\mathrm{GP}} \bigg]\! \begin{pmatrix}
 u_t\\ v_t
 \end{pmatrix}\!\!\Big\rangle_{\!L^2(\mathbb{R}^3)\otimes\mathbb{C}^2}\\
 &\qquad\quad +\Big\langle\!\!\begin{pmatrix}
 u_t \\   v_t
 \end{pmatrix}\!,\bigg(\,\frac{\ud}{\ud t}\begin{pmatrix}
 4\pi a\big( |u_t|^2+|v_t|^2\big) &0\\0&4\pi a\big( |u_t|^2+|v_t|^2\big)
 \end{pmatrix} \!\!\bigg)\! \begin{pmatrix}
 u_t\\ v_t
 \end{pmatrix}\!\!\Big\rangle_{\!L^2(\mathbb{R}^3)\otimes\mathbb{C}^2}\\
  &\qquad\quad+\Big\langle\!\!\begin{pmatrix}
 u_t \\   v_t
 \end{pmatrix}\!,\dot S(t)\! \begin{pmatrix}
 u_t\\ v_t
 \end{pmatrix}\!\!\Big\rangle_{\!L^2(\mathbb{R}^3)\otimes\mathbb{C}^2} \,.
\end{split}
 \end{equation}
For the second summand of \eqref{eq:time_derivative_partial} we compute, by the Leibniz rule,
\[
\begin{split}
\Big\langle&\!\!\begin{pmatrix}
u_t \\   v_t
\end{pmatrix}\!,\bigg(\,\frac{\ud}{\ud t}\begin{pmatrix}
|u_t|^2+|v_t|^2 &0\\0& |u_t|^2+|v_t|^2
\end{pmatrix} \!\!\bigg)\! \begin{pmatrix}
u_t\\ v_t
\end{pmatrix}\!\!\Big\rangle_{\!L^2(\mathbb{R}^3)\otimes\mathbb{C}^2}\\
=&\;\big\langle u_t,\Big[\big(\de_t\overline u_t\big)u_t+\overline u_t\big(\de_tu_t\big)\Big] u_t\big\rangle_{L^2(\mathbb{R}^3)}\;+\;\big\langle u_t,\Big[\big(\de_t\overline v_t\big)v_t+\overline v_t\big(\de_tv_t\big)\Big] u_t\big\rangle_{L^2(\mathbb{R}^3)}\\
&+\;\big\langle v_t,\Big[\big(\de_t\overline u_t\big)u_t+\overline u_t\big(\de_tu_t\big)\Big] v_t\big\rangle_{L^2(\mathbb{R}^3)}\;+\;\big\langle v_t,\Big[\big(\de_t\overline v_t\big)v_t+\overline v_t\big(\de_tv_t\big)\Big] v_t\big\rangle_{L^2(\mathbb{R}^3)}\\
\;=&\;\big\langle \de_t u_t,|u_t|^2u_t\big\rangle_{L^2(\mathbb{R}^3)}\;+\;\big\langle  u_t,|u_t|^2\big(\de_tu_t\big)\big\rangle_{L^2(\mathbb{R}^3)}+\big\langle \de_t v_t,|u_t|^2v_t\big\rangle_{L^2(\mathbb{R}^3)}\;+\;\big\langle  v_t,|u_t|^2\big(\de_tv_t\big)\big\rangle_{L^2(\mathbb{R}^3)}\\
&+\;\big\langle \de_t u_t,|v_t|^2u_t\big\rangle_{L^2(\mathbb{R}^3)}\;+\;\big\langle  u_t,|v_t|^2\big(\de_tu_t\big)\big\rangle_{L^2(\mathbb{R}^3)}+\big\langle \de_t v_t,|v_t|^2v_t\big\rangle_{L^2(\mathbb{R}^3)}\;+\;\big\langle  v_t,|v_t|^2\big(\de_tv_t\big)\big\rangle_{L^2(\mathbb{R}^3)}.
\end{split}
\] 
Bringing the latter expression back to spinorial form gives
\[
\begin{split}
\Big\langle\!\!\begin{pmatrix}
u_t \\   v_t
\end{pmatrix}\!,\bigg(\,\frac{\ud}{\ud t}&\begin{pmatrix}
|u_t|^2+|v_t|^2 &0\\0& |u_t|^2+|v_t|^2
\end{pmatrix} \!\!\bigg)\! \begin{pmatrix}
u_t\\ v_t
\end{pmatrix}\!\!\Big\rangle_{\!L^2(\mathbb{R}^3)\otimes\mathbb{C}^2}\\
=&\;\,\Big\langle\frac{\ud}{\ud t}\begin{pmatrix}
u_t \\   v_t
\end{pmatrix},\begin{pmatrix}
|u_t|^2+|v_t|^2 &0\\0& |u_t|^2+|v_t|^2
\end{pmatrix} \! \begin{pmatrix}
u_t\\ v_t
\end{pmatrix}\!\!\Big\rangle_{\!L^2(\mathbb{R}^3)\otimes\mathbb{C}^2}\\
&\;+\,\Big\langle\!\!\begin{pmatrix}
u_t\\ v_t
\end{pmatrix},\begin{pmatrix}
|u_t|^2+|v_t|^2 &0\\0& |u_t|^2+|v_t|^2
\end{pmatrix} \frac{\ud}{\ud t}\begin{pmatrix}
u_t \\   v_t
\end{pmatrix}\!\!\Big\rangle_{\!L^2(\mathbb{R}^3)\otimes\mathbb{C}^2},
\end{split}
\]
which, since by \eqref{eq:coupled} and \eqref{eq:spin_hamiltonian} the time derivatives produce $-\ii h^{\mathrm{GP}}$, yields
\[
\begin{split}
\Big\langle\!\!\begin{pmatrix}
u_t \\   v_t
\end{pmatrix}\!,\bigg(\,\frac{\ud}{\ud t}&\begin{pmatrix}
|u_t|^2+|v_t|^2 &0\\0& |u_t|^2+|v_t|^2
\end{pmatrix} \!\!\bigg)\! \begin{pmatrix}
u_t\\ v_t
\end{pmatrix}\!\!\Big\rangle_{\!L^2(\mathbb{R}^3)\otimes\mathbb{C}^2}\\
& \quad=\;-\,\ii\,\Big\langle\!\!\begin{pmatrix}\label{eq:appendix_2}
u_t \\   v_t
\end{pmatrix}\!,\bigg[\,\begin{pmatrix}
|u_t|^2+|v_t|^2 &0\\0& |u_t|^2+|v_t|^2
\end{pmatrix}, h^{\mathrm{GP}} \bigg]\! \begin{pmatrix}
u_t\\ v_t
\end{pmatrix}\!\!\Big\rangle_{\!L^2(\mathbb{R}^3)\otimes\mathbb{C}^2}.	
\end{split}
\]
This identity shows that an exact cancellation takes place between the first two summands of \eqref{eq:time_derivative_partial}, and using also $[h^{\mathrm{GP}},h^{\mathrm{GP}}]=\mathbb{O}$) one gets
\[
\frac{\ud}{\ud t}\mathcal{E}^{\mathrm{GP}}[u_t,v_t]\;=\;\Big\langle\!\!\begin{pmatrix}
u_t \\   v_t
\end{pmatrix}\!,\dot S(t)\! \begin{pmatrix}
u_t\\ v_t
\end{pmatrix}\!\!\Big\rangle_{\!L^2(\mathbb{R}^3)\otimes\mathbb{C}^2}.
\]
Comparing the quantities $\frac{\ud}{\ud t}\mathcal{E}_N[\Psi_{N,t}]$ and $\frac{\ud}{\ud t}\mathcal{E}^{\mathrm{GP}}[u_t,v_t]$ computed above with \eqref{eq:def_dN_a} yields finally the conclusion \eqref{eq:enest}.
\end{proof}

Next, in order to establish the identity \eqref{eq:ddelta_abcdef} of Proposition \ref{prop:alphaestimate}, let us single out the following fact.

\begin{lemma} \label{lemma:ddt}Under the assumptions of Theorem \ref{theorem:main}, one has
	\[
	\frac{\ud}{\ud t}\langle\Psi_{N,t},\widehat{m}_t \Psi_{N,t}\rangle_{\!\cH_N}\;=\;-N(N-1)\,{\mathfrak{Im}}\,\langle\Psi_{N,t},Z_{N,t}(x_1,x_2)\,R_{(12),t}\,\Psi_{N,t}\rangle_{\!\cH_N}\,,\qquad t\geqslant 0\,,
	\]
	with $Z_{N,t}(x_1,x_2)$ defined in \eqref{eq:defz}.
\end{lemma}

\begin{proof} 
First, owing to \eqref{eq:coupled} and to the definition \eqref{eq:spin_hamiltonian} of $h^{\mathrm{GP}}$, 
	\begin{equation}\label{eq:derivative_projector_p}
\begin{split}
	\frac{\ud}{\ud t}\,p_t\;&=\;\frac{\ud}{\ud t}\bigg(\:\Big| \!\!\begin{pmatrix} u_t \\ v_t\end{pmatrix}\!\!\Big\rangle\Big\langle\!\! \begin{pmatrix} u_t \\ v_t\end{pmatrix}\!\!\Big|\:\bigg)\;=\; -\ii\,[h^{\mathrm{GP}}(t),p_t] \\
\frac{\ud}{\ud t}\,q_t\;&=\;-\frac{\ud}{\ud t}\,p_t\;=\;\ii\,[h^{\mathrm{GP}}(t),p_t] \;=\;-\ii\,[h^{\mathrm{GP}}(t),q_t]\,,
\end{split}
	\end{equation}
whence, differentiating in time in \eqref{eq:defPk},
	\begin{equation}\label{eq:derivative_projector_P}
	\frac{\ud}{\ud t}\,P_k\;=\;-\ii\Big[\sum_{j=1}^Nh^{\mathrm{GP}}_j,P_k\Big],\qquad k\in\{0,\dots,N\}\,.
	\end{equation}
When the time derivative of $\langle\Psi_{N,t},\widehat{m}_t \Psi_{N,t}\rangle_{\!\cH_N}$ hits the $\Psi_{N,t}$'s, this produces a commutator term $[H_N,\widehat{m}_t]$, owing to \eqref{eq:Cauchy_problem}, whereas when the time derivatives hits each $P_k$ in $\widehat{m}_t$, this produces a commutator term of the form \eqref{eq:derivative_projector_P}. Thus,
	\[
	\frac{\ud}{\ud t}\,\langle\Psi_{N,t},\widehat{m}_t \Psi_{N,t}\rangle_{\!\cH_N}\;=\;\ii \,\big\langle\Psi_{N,t},\big[H_N-\sum_{j=1}^N h^{\mathrm{GP}}_j,\,\widehat{m}_t\big]\Psi_{N,t}\big\rangle_{\!\cH_N}\,.
	\]
In the r.h.s.~above an exact cancellation occurs between the terms $\sum_{j=1}^N(-\Delta_{x_j})+\sum_{j=1}^N S(x_j,t)$ given by $H_N$ and the same terms given by $\sum_{j=1}^N h^{\mathrm{GP}}_j$, and what remains is
	\[
	\frac{\ud}{\ud t}\,\langle\Psi_{N,t},\widehat{m}_t \Psi_{N,t}\rangle_{\!\cH_N}\;=\;\ii \,\big\langle\Psi_{N,t},\big[\sum_{i<j}^NV_N(x_i-x_j)-\sum_{i=1}^N8\pi a(|u_t(x_i)|^2+|v_t(x_i)|^2,\,\widehat{m}_t\big]\Psi_{N,t}\big\rangle_{\!\cH_N}\,.
	\]
Because of the bosonic symmetry of $\Psi_{N,t}$ and $\widehat{m}_t$, the identity above reads also
	\[
	\begin{split}
		&\frac{\ud}{\ud t}\,\langle\Psi_{N,t},\widehat{m}_t \Psi_{N,t}\rangle_{\!\cH_N}\;= \\
&\;\; =\;{\textstyle\frac{1}{2}}\,\ii\,N(N-1)\big\langle\Psi_{N,t},\big[V_N(x_1-x_2)\\
&\qquad\qquad\qquad\qquad-{\textstyle\frac{8\pi a}{N-1}}(|u_t(x_1)|^2+|v_t(x_1)|^2)+|u_t(x_2)|^2+|v_t(x_2)|^2),\,\widehat{m}_t\big]\Psi_{N,t}\big\rangle_{\!\cH_N}\\
	&\;\;=\;{\textstyle\frac{1}{2}}\,\ii\,N(N-1)\langle\Psi_{N,t},[Z_{N,t}(x_1,x_2),\,\widehat{m}_t]\Psi_{N,t}\rangle_{\!\cH_N}\,,
	\end{split}
	\]
	where $Z_{N,t}(x_1,x_2)$ is defined in \eqref{eq:defz}. Last,
	\[\begin{split}
	\frac{\ud}{\ud t}\,\langle\Psi_{N,t},\widehat{m}_t \Psi_{N,t}\rangle_{\!\cH_N}\;&=\;{\textstyle\frac{1}{2}}\,\ii\,N(N-1)\,\langle\Psi_{N,t},[Z_{N,t}(x_1,x_2),R_{(12),t}]\,\Psi_{N,t}\rangle_{\!\cH_N}\\
	&=\;-N(N-1)\,\mathfrak{Im}\,\langle\Psi_{N,t},Z_{N,t}(x_1,x_2)R_{(12),t}\,\Psi_{N,t}\rangle_{\!\cH_N}\,,
	\end{split}
	\]
where in the first identity we used \eqref{eq:commut_r} of Lemma \ref{lemma:tools}(iii) and in the second identity we used the property $\langle\varphi,[A,B]\varphi\rangle\;=\;2\,\ii\,\mathfrak{Im}\langle\varphi, AB\varphi\rangle$ of bounded symmetric operators $A$ and $B$.
\end{proof}



\begin{proof}[Proof of Proposition \ref{prop:alphaestimate}]
It follows at once from the definition \eqref{eq:def_deltaN} of $\delta_N$ and from Lemma \ref{lemma:ddt} above that
\[
 \begin{split}
  \frac{\ud}{\ud t}\delta_N(t)\;&=\;-N(N-1)\,{\mathfrak{Im}}\,\langle\Psi_{N,t},Z_{N,t}(x_1,x_2)\,R_{(12),t}\,\Psi_{N,t}\rangle_{\!\cH_N} \\
&\qquad -N(N-1)\,\frac{\ud}{\ud t}\mathfrak{Re}\langle\Psi_{N,t},g_\beta(x_1-x_2)\,R_{(12),t}\,\Psi_{N,t}\rangle_{\!\cH_N} \,.
 \end{split}
\]
In the r.h.s.~above the time derivative can either hit the $\Psi_{N,t}$'s, thus producing $H_N\Psi_{N,t}$ via \eqref{eq:Cauchy_problem}, or the operator $R_{(12),t}$: in the latter case, we see from the definition \eqref{eq:defnr} of $R_{(ij)}$ and from \eqref{eq:derivative_projector_p} that
	\[
	\frac{\ud}{\ud t}\,R_{(k\ell)}\;=\;-\ii\Big[\sum_{j=1}^Nh^{\mathrm{GP}}_j,R_{(k\ell)}\Big]\,.
	\]
Therefore,
	\[
	\begin{split}
	\frac{\ud}{\ud t}\delta_N(t)\;&=\;-N(N-1)\,{\mathfrak{Im}}\,\langle\Psi_{N,t},Z_{N,t}(x_1,x_2)\,R_{(12),t}\,\Psi_{N,t}\rangle_{\!\cH_N}\\
	&\qquad-\,N(N-1)\,\mathfrak{Im}\,\big\langle\Psi_{N,t},g_\beta(x_1-x_2)\big[\sum_{j=1}^N h^{\mathrm{GP}}_j,R_{(12),t}\big]\Psi_{N,t}\big\rangle_{\!\cH_N}\\
	&\qquad-\,N(N-1)\,\mathfrak{Im}\,\langle\Psi_{N,t},g_\beta(x_1-x_2)R_{(12,t)}\,H_N\,\Psi_{N,t}\rangle_{\!\cH_N}\\
	&\qquad+\,N(N-1)\,\mathfrak{Im}\,\langle\Psi_{N,t},H_N \,g_\beta(x_1-x_2)R_{(12),t}\,\Psi_{N,t}\rangle_{\!\cH_N}\\
	&=\;-N(N-1)\,{\mathfrak{Im}}\,\langle\Psi_{N,t},Z_{N,t}(x_1,x_2)\,R_{(12),t}\,\Psi_{N,t}\rangle_{\!\cH_N}\\
	&\qquad +N(N-1)\,\mathfrak{Im}\,\langle\Psi_{N,t},[\,H_N,g_\beta(x_1-x_2)]\,R_{(12),t}\,\Psi_{N,t}\rangle_{\!\cH_N}\\
	&\qquad+N(N-1)\,\mathfrak{Im}\,\langle\Psi_{N,t},g_\beta(x_1-x_2)\big[H_N-\sum_{j=1}^Nh^{\mathrm{GP}}_j,R_{(12),t}\big]\,\Psi_{N,t}\rangle_{\!\cH_N}\,.
	\end{split}
	\]
In the last summand above, both $g_\beta(x_1-x_2)$ and $R_{(12)}$ break the full bosonic symmetry: as a consequence, $H_N-\sum_{j=1}^Nh^{GP}_j$ produces several terms, depending on the presence or absence of the variables $x_1$ and $x_2$. We find 
	\begin{align*}
	&\frac{\ud}{\ud t}\delta_N(t)\;=\;-N(N-1)\,{\mathfrak{Im}}\,\langle\Psi_{N,t},Z_{N,t}(x_1,x_2)\,R_{(12),t}\,\Psi_{N,t}\rangle_{\!\cH_N} \\
	&\;+N(N-1)\,\mathfrak{Im}\,\langle\Psi_{N,t},[\,H_N,g_\beta(x_1-x_2)]\,R_{(12),t}\,\Psi_{N,t}\rangle_{\!\cH_N} \\
	&\;+N(N-1)\mathfrak{Im}\langle\Psi_N,g_\beta(x_1-x_2)[Z(x_1,x_2),R_{(12),t}]\Psi_N\rangle_{\!\cH_N} \\
	&\;+N(N-2)\mathfrak{Im}\langle\Psi_{N,t},g_\beta(x_1-x_2)\big(|u_t(x_1)|^2+|u_t(x_2)|^2+|v_t(x_1)|^2+|v_t(x_2)|^2,R_{(12),t} \big)\Psi_{N,t}\rangle_{\!\cH_N} \\
&\;+N(N-1)(N-2)\mathfrak{Im}\,\langle\Psi_{N,t},g_\beta(x_1-x_2)[V_N(x_1-x_3)+8\pi a(|u_t(x_3)|^2+|v_t(x_3)|^2),R_{(12),t} ]\Psi_{N,t}\rangle_{\!\cH_N} \\
&\;+{\textstyle\frac{1}{2}}\,N(N-1)(N-2)(N-3)\mathfrak{Im}\,\langle\Psi_{N,t},g_\beta(x_1-x_2)[V_N(x_3-x_4),R_{(12),t} ]\Psi_{N,t}\rangle_{\!\cH_N}\,. 
	\end{align*}
The last \emph{three} summands in the r.h.s.~above are recognised to be, respectively, $\delta_N^{(f)}(t)$, $\delta_N^{(d)}(t)$, and $\delta_N^{(e)}(t)$, whence
\[
 \begin{split}
  \frac{\ud}{\ud t}\delta_N(t)\;=&\;\delta_N^{(d)}(t) + \delta_N^{(e)}(t)+\delta_N^{(f)}(t)+N(N-1)\,\mathfrak{Im}\,\langle\Psi_{N,t},[\,H_N,g_\beta(x_1-x_2)]\,R_{(12),t}\,\Psi_{N,t}\rangle_{\!\cH_N} \\
 &\;-N(N-1)\,\mathfrak{Im}\,\langle\Psi_{N,t},(1-g_\beta(x_1-x_2))\, Z_{N,t}(x_1,x_2)\,R_{(12),t}\,\Psi_{N,t}\rangle_{\!\cH_N} \\
&\;-N(N-1)\,\mathfrak{Im}\,\langle\Psi_{N,t},g_\beta(x_1-x_2)\,R_{(12),t}\, Z_{N,t}(x_1,x_2)\,\Psi_{N,t}\rangle_{\!\cH_N} \,.
 \end{split}
\]
By means of the identity
	\begin{equation*} 
	(1-g_\beta(x_1-x_2))Z_{N,t}(x_1,x_2)=Z_{N,t}^{(\beta)}(x_1,x_2)+(V_N(x_1-x_2)-W_\beta(x_1-x_2))f_\beta(x_1-x_2),
	\end{equation*}
that follows from \eqref{eq:defng}, \eqref{eq:defz_beta}, and \eqref{eq:defz}, the above expression for $\frac{\ud}{\ud t}\delta_N(t)$ takes the form 
\[
 \begin{split}
  \frac{\ud}{\ud t}\delta_N(t)\;&=\;\delta_N^{(d)}(t) + \delta_N^{(e)}(t)+\delta_N^{(f)}(t) \\
&\qquad -N(N-1)\,\mathfrak{Im}\Big(\langle\Psi_{N,t},Z^{(\beta)}_{N,t}(x_1,x_2)R_{(12),t}\Psi_{N,t}\rangle_{\!\cH_N}+ \\
&\qquad \qquad\qquad\qquad\qquad+ \langle\Psi_{N,t},g_\beta(x_1-x_2)R_{(12),t}\,Z_{N,t}(x_1,x_2)\Psi_{N,t}\rangle_{\!\cH_N}\Big)\,, \\
&\qquad -N(N-1)\,\mathfrak{Im}\,\langle\Psi_{N,t},(V_N(x_1-x_2)-W_\beta(x_1-x_2))f_\beta(x_1-x_2)R_{(12),t}\,\Psi_{N,t}\rangle_{\!\cH_N}\\
&\qquad +N(N-1)\,\mathfrak{Im}\,\langle\Psi_{N,t},[\,H_N,g_\beta(x_1-x_2)]\,R_{(12),t}\,\Psi_{N,t}\rangle_{\!\cH_N} \\
& = \;\delta_N^{(b)}(t) + \delta_N^{(d)}(t) + \delta_N^{(e)}(t)+\delta_N^{(f)}(t) \\
&\qquad -N(N-1)\,\mathfrak{Im}\,\langle\Psi_{N,t},(V_N(x_1-x_2)-W_\beta(x_1-x_2))f_\beta(x_1-x_2)R_{(12),t}\,\Psi_{N,t}\rangle_{\!\cH_N}\\
&\qquad +N(N-1)\,\mathfrak{Im}\,\langle\Psi_{N,t},[\,H_N,g_\beta(x_1-x_2)]\,R_{(12),t}\,\Psi_{N,t}\rangle_{\!\cH_N}\,.
 \end{split}
\]
Last, let us focus on the last two summands in the r.h.s.~above. Precisely at this level a cancellation occurs in which the difference $V_N-W_\beta$ is controlled by the commutator $[H_N,g_\beta]$, at the cost of the further term $\delta_N^{(c)}$ that is going to appear in a moment. We compute
\[
 \begin{split}
  [H_N&,g_\beta(x_1-x_2)]\;=\;-[H_N,f_\beta(x_1-x_2)]\;=\;[\Delta_{x_1}+\Delta_{x_2},f_\beta(x_1-x_2)] \\
 &=\;(\Delta_{x_1}+\Delta_{x_2})f_\beta(x_1-x_2)+2(\nabla_{x_1}f_\beta(x_1-x_2))\nabla_{x_1}+2(\nabla_{x_2}f_\beta(x_1-x_2))\nabla_{x_2} \\
&=\;(V_N(x_1-x_2)-W_\beta(x_1-x_2))f_\beta(x_1-x_2)\\
&\quad\,\,\,\,-2(\nabla_{x_1}g_\beta(x_1-x_2))\nabla_{x_1}-2(\nabla_{x_2}g_\beta(x_1-x_2))\nabla_{x_2}
 \end{split}
\]
having used \eqref{eq:defng} in the first identity and the zero-energy scattering equation \eqref{eq:zesp-beta} in the last one. We thus see that the $(V_N-W_\beta)f_\beta$-term gets cancelled out in the above expression for $\frac{\ud}{\ud t}\delta_N(t)$, whereas the $(\nabla g)\nabla$-terms produce precisely the expression \eqref{eq:def_dN_c} for $\delta_N^{(c)}$. The conclusion is
\[
 \frac{\ud}{\ud t}\delta_N(t)\;=\;\delta_N^{(b)}(t) +\delta_N^{(c)}(t) +\delta_N^{(d)}(t) + \delta_N^{(e)}(t)+\delta_N^{(f)}(t)\,,
\]
which completes the proof.
\end{proof}

Last, we establish the bounds \eqref{eq:bound_dN}.

\begin{proof}[Proof of Proposition \ref{prop:each_g_abcdef}]
Let us discuss each case $\delta_N^{(j)}$, $j\in\{a,b,c,d,e,f\}$ separately.

\medskip

\noindent\textbf{Term $\delta_N^{(a)}$.} We recall that
\[
\delta_N^{(a)}(t)\;=\;\langle\Psi_{N,t}, \dot S(x_1,t)\Psi_{N,t}\rangle_{\!\cH_N}-\Big\langle \begin{pmatrix}u_t\\v_t\end{pmatrix} ,\,\dot S(t)\begin{pmatrix}u_t\\v_t\end{pmatrix}\Big\rangle_{L^2(\mathbb{R}^3)\otimes\mathbb{C}^2}\,.
\]
Inserting $\mathbbm{1}=p_t+q_t$ into the first summand yields
\[
\begin{split}
\delta_N^{(a)}(t)\;=\;&\langle\Psi_{N,t}, (p_t)_1\dot S(x_1,t)(p_t)_1\Psi_{N,t}\rangle_{\!\cH_N}+\langle\Psi_{N,t}, (q_t)_1\dot S(x_1,t)(q_t)_1\Psi_{N,t}\rangle_{\!\cH_N}\\
&\;+2 \mathfrak{Re} \langle\Psi_{N,t}, (q_t)_1\dot S(x_1,t)(p_t)_1\Psi_{N,t}\rangle_{\!\cH_N}-\Big\langle \begin{pmatrix}u_t\\v_t\end{pmatrix} ,\,\dot S(x,t)\begin{pmatrix}u_t\\v_t\end{pmatrix}\Big\rangle_{L^2(\mathbb{R}^3)\otimes\mathbb{C}^2}\,.
\end{split}
\]
We then use the identity
\begin{equation}\label{eq:starid}
p_1A(x_1)p_1=p_1\Big\langle\begin{pmatrix}u\\v\end{pmatrix} ,\,A(x)\begin{pmatrix}u\\v\end{pmatrix}\Big\rangle_{L^2(\mathbb{R}^3)\otimes\mathbb{C}^2}\,,
\end{equation}
which is valid for any 2x2 operator-valued matrix $A(x)$, the $L^\infty$-boundedness of $\dot S$ (see assumption (A1)), and the invertibility of $\widehat{n}_t^{1/2}$ on the range of $(q_t)_1$ (i.e., $\mathbbm{1}_{\operatorname{Ran}(q_t)_1}=\widehat{n}_t^{-1/2}\widehat{n}_t^{1/2}$), so as to obtain
\begin{align} \label{eq:delta_a_partial1}
\delta_N^{(a)}(t)\;\leqslant\;|\delta_N^{(a)}(t)|\;\leqslant\;&\Big(1-\|(p_t)_1\Psi_{N,t}\|^2\Big)\,\bigg|\Big\langle \begin{pmatrix}u_t\\v_t\end{pmatrix} ,\,\dot S(t)\begin{pmatrix}u_t\\v_t\end{pmatrix}\Big\rangle_{L^2(\mathbb{R}^3)\otimes\mathbb{C}^2}\bigg|\\
&+\|\dot S\|_{L^\infty_t L^\infty_x}\|(q_t)_1\Psi_{N,t}\|^2 \label{eq:delta_a_partial2}\\
&+2 \big|	\langle\Psi_{N,t}, \hat n^{-1/2}_t\widehat n^{1/2}_t\,(q_t)_1\dot S(x_1,t)(p_t)_1\Psi_{N,t}\rangle_{\cH_N}\big|.\label{eq:delta_a_partial3}
\end{align}
The term \eqref{eq:delta_a_partial1} is controlled by $\|\dot S\|_{L^\infty_tL^\infty_x}\,\|(q_t)_1\Psi_{N,t}\|^2$ (indeed $1-\|(p_t)_1\Psi_{N,t}\|^2=\|(q_t)_1\Psi_{N,t}\|^2$). In the term \eqref{eq:delta_a_partial3} we shift $\widehat{n}_t^{1/2}$ to $\widehat{n}_{1,t}^{1/2}$ by means of \eqref{eq:commutation}. This and a Schwarz inequality yield
\begin{equation} \label{eq:delta_a_partial4}
\begin{split}
&|\delta_N^{(a)}(t)|\;\leqslant\;2\,\|\dot S\|_{L^\infty_tL^\infty_x}\,\bigg(\|(q_t)_1\Psi_{N,t}\|^2+\Big|\langle\Psi_{N,t}, \widehat n^{-1/2}_t\,(q_t)_1\dot S(x_1,t)\widehat n^{1/2}_{1,t}(p_t)_1\Psi_{N,t}\rangle_{\cH_N}\Big|\bigg)\\
&\leqslant\;\widetilde C\;\bigg(\|(q_t)_1\Psi_{N,t}\|^2+\|\widehat n^{-1/2}_t(q_t)_1\Psi_{N,t}\|\sqrt{\langle\Psi_{N,t},\widehat n_{1,t}^{1/2}(p_t)_1\dot S(x_1,t)^2(p_t)_1\widehat n_{1,t}^{1/2}\Psi_{N,t}\rangle_{\cH_N}}\bigg),
\end{split}
\end{equation}
for some constant $\widetilde C>0$. Moreover, owing to \eqref{eq:starid},
\[
\|p_1\dot S(x_1,t)^2p_1\|_{\mathrm{op}}\;=\;\|p_1\|_{\mathrm{op}}\;\Big|\Big\langle \begin{pmatrix}u\\v\end{pmatrix} ,\,\dot S(x,t)^2\begin{pmatrix}u\\v\end{pmatrix}\Big\rangle_{L^2(\mathbb{R}^3)\otimes\mathbb{C}^2}\Big|\;\leqslant\;\|\dot S\|_{L^\infty_t L^\infty_x}^2\,,
\]
and hence
\[
\sqrt{\langle\Psi_{N,t},\widehat n_{1,t}^{1/2}(p_t)_1\dot S(x_1,t)^2(p_t)_1\widehat n_{1,t}^{1/2}\Psi_{N,t}\rangle_{\cH_N}}\;\leqslant\;\|\dot S\|_{L^\infty_t L^\infty_x}\|\widehat n_{1,t}^{1/2}\Psi_{N,t}\|\,;
\]
also,
\[
\|(q_t)_1\Psi_{N,t}\|^2\;=\;\|\widehat n^{1/2}_t\Psi_{N,t}\|^2
\]
due to \eqref{eq:relation_q_n}, and
\[
\|\widehat n^{-1/2}_t(q_t)_1\Psi_{N,t}\|\;\leqslant\;\|\widehat n^{1/2}_t\Psi_{N,t}\|
\]
due to Lemma \ref{lemma:nq}. These facts, together with the operator bound $\widehat{n}_1\;\leqslant\;\widehat{n}+N^{-1/2}\mathbbm{1}$, give
\[
\begin{split}
|\delta_N^{(a)}(t)|\;&\leqslant\;\widehat{C}\big(\|\widehat n^{1/2}_t\Psi_{N,t}\|^2+\|\widehat n^{1/2}_t\Psi_{N,t}\|\|\widehat n^{1/2}_{1,t}\Psi_{N,t}\|\big)\\
&\leqslant\;\widehat{C}\bigg(\|\widehat n^{1/2}_t\Psi_{N,t}\|^2+\|\widehat n^{1/2}_t\Psi_{N,t}\|\sqrt{\|\widehat{n}^{1/2}_t\Psi_{N,t}\|^2+\frac{1}{\sqrt{N}}}\bigg)\\
&\leqslant\;\widehat{C}\Big(\|\widehat n^{1/2}_t\Psi_{N,t}\|^2+\|\widehat n^{1/2}_t\Psi_{N,t}\|^2+\frac{1}{\,N^{1/4}}\|\widehat{n}^{1/2}_t\Psi_{N,t}\|\Big)\;\leqslant\;C\Big(\|\widehat n^{1/2}_t\Psi_{N,t}\|^2+\frac{1}{\sqrt{N}}\Big)
\end{split}
\]
for some constants $\widehat C,C>0$. Last, applying \eqref{eq:m-n-n2}, we conclude
\[
|\delta_N^{(a)}(t)|\;\leqslant\;C\Big(\alpha_N^<(t)+\frac{1}{\sqrt{N}}\Big).
\] 

\bigskip

\noindent\textbf{Term $\delta_N^{(b)}$.} This term is crucial, for it is the only one containing, through $Z_{N,t}^{(\beta)}$, the actual difference between $W_\beta$ and the effective non-linear potential -- and \emph{this} difference is controllable, unlike the analogous difference with $V_N$ in place of $W_\beta$. Concerning the $Z_{N,t}$-term in $\delta_N^{(b)}$, which alone would not be controllable either, the nearby $g_\beta$ allows for an efficient estimate too.

We start with splitting
\begin{align}
\delta_N^{(b)}(t)\;=&\;-N(N-1)\mathfrak{Im}\langle\Psi_{N,t},g_\beta(x_1-x_2)R_{(12),t}\;Z_{N,t}(x_1,x_2)\Psi_{N,t}\rangle_{\!\cH_N} \label{eq:gamma_b2}\\
&\;-N(N-1)\mathfrak{Im}\langle\Psi_{N,t},\,Z_{N,t}^{(\beta)}(x_1,x_2)R_{(12),t}\Psi_{N,t}\rangle_{\!\cH_N} \label{eq:gamma_b1}\;.
\end{align}

In order to bound \eqref{eq:gamma_b2} we observe that from \eqref{eq:defnr} each summand of $R_{(12),t}$ contains at least one $p_t$, either in the variable $x_1$ or $x_2$. Since \eqref{eq:gamma_b2} is symmetric under exchange of $x_1\leftrightarrow x_2$, it follows that  $(p_t)_1(q_t)_2$ and $(q_t)_1(p_t)_2$ give the same contribution. Then
\[
\begin{split}
|\eqref{eq:gamma_b2}|\;\leqslant\;2\,N^2\|g_\beta(x_1-x_2)(p_t)_1\|_{\mathrm{op}}\big(\|\widehat{m}^a_t\|_{\mathrm{op}}+\|\widehat{m}^b_t\|_{\mathrm{op}}\big)\|(p_t)_1 Z_{N,t}(x_1,x_2)\Psi_{N,t}\|\,,
\end{split}
\]
having used the $p$'s coming from $R_{(12)}$ to multiply both $g_\beta$ and $Z_{N,t}(x_1,x_2)$. By means of the bounds \eqref{eq:ma}, \eqref{eq:mb}, \eqref{eq:dressed2}, \eqref{eq:dressedpotential}, and \eqref{eq:g2}, and the fact that the most singular contribution to $Z_{N,t}$ is given by $V_N$, we obtain
\begin{equation} \label{eq:gamma_b2_final}
|\eqref{eq:gamma_b2}|\;\leqslant\; \widehat{C}(t)\, N^{1+\xi}\|g_\beta\|_2\|(p_t)_1V_N(x_1,x_2)\Psi_{N,t}\|\;\leqslant\; \widetilde{C}(t)\,N^{-1-\frac{\beta}{2}+\xi},
\end{equation}
for suitable $\widehat{C}(t),\widetilde{C}(t)>0$ that depend on $\|u_t\|_{\infty}$ and $\|v_t\|_{\infty}$ but not on $N$. $Z_{N,t}$ contains also terms depending on $|u_t|^2$ and $|v_t|^2$ that are bounded the same way.

The summand \eqref{eq:gamma_b1}, in turn, is split as 
\begin{align}
&|\eqref{eq:gamma_b1}|\;\leqslant\;N^2\,\big|\langle\Psi_{N,t},\big(W_\beta(x_1-x_2)f_\beta(x_1-x_2)\nonumber\\
&\qquad\qquad-{\textstyle \frac{8\pi a}{N-1}}\,(|u_t(x_1)|^2+|v_t(x_1)|^2+|u_t(x_2)|^2+|v_t(x_2)|^2)\big)R_{(12),t}\Psi_{N,t}\rangle_{\!\cH_N}\big| \label{eq:gamma_b3}\\
&+ N \,\big|\langle\Psi_{N,t},8\pi a (|u_t(x_1)|^2+|v_t(x_1)|^2+|u_t(x_2)|^2+|v_t(x_2)|^2)g_\beta(x_1-x_2)R_{(12),t}\Psi_{N,t}\rangle_{\!\cH_N}\big|, \label{eq:gamma_b4}
\end{align}
having used $f_\beta=1-g_\beta$ and the definition \eqref{eq:defz_beta} of $Z_{N,t}^{(\beta)}$.
We now recognise that the summand \eqref{eq:gamma_b3} can be estimated by means of a general result from \cite{Pickl-RMP-2015} which for convenience we state in Lemma \ref{lemma:appendix_2}. Indeed, the potential $W_\beta f_\beta$ does satisfy the conditions \eqref{eq:Wtilde_1}, \eqref{eq:Wtilde_2}, \eqref{eq:Wtilde_3} of Lemma \ref{lemma:appendix_2}: condition \eqref{eq:Wtilde_1} is obvious from \eqref{eq:def_W_beta} and \eqref{eq:R_beta}; condition \eqref{eq:Wtilde_2} follows from \eqref{eq:def_W_beta} and the uniform boundedness of $f_\beta$; condition \eqref{eq:Wtilde_3} is explicitly checked in \cite[Lemma 5.1]{Pickl-RMP-2015}. Condition \eqref{eq:Phi_N} is satisfied too, where the vector $\Phi_N$ of Lemma \ref{lemma:appendix_2} is, in our case, precisely $\Psi_{N,t}$. Indeed, due to the positivity of $V_N$,
\[
\begin{split}
\mathcal{E}_N[\Psi_{N,t}]&\;=\;\|\nabla_1\Psi_{N,t}\|^2+\langle \Psi_{N,t}, S (x_1,t) \Psi_{N,t}\rangle_{\cH_N}+\frac{1}{2}(N-1)\langle\Psi_{N,t}, V_N(x_1-x_2)\Psi_{N,t}\rangle_{\cH_N}\\
&\;\geqslant\;\|\nabla_1\Psi_{N,t}\|^2+\langle \Psi_{N,t}, S (x_1,t) \Psi_{N,t}\rangle_{\cH_N}\,,
\end{split}
\] 
whence
\begin{equation} \label{eq:estimate_H1}
\|\nabla_1\Psi_{N,t}\|^2\;\leqslant\;\big|\mathcal{E}_N[\Psi_{N,t}]\big|+\|S\|_{L^\infty_t\,L^\infty_x}\,.
\end{equation}
On the other hand, integrating the bound 
\[
\frac{\ud}{\ud t}\mathcal{E}_N[\Psi_{N,t}]\;\leqslant\; \|\dot S\|_{L^\infty_t\,L^\infty_x}
\]
(see \eqref{eq:derivative_linear} above), yields
\begin{equation} \label{eq:estimate_energy}
\mathcal{E}_N[\Psi_{N,t}]\;\leqslant\;G(t)
\end{equation}
for some positive and $N$-\emph{independent} function $G(t)$. Thus, \eqref{eq:estimate_H1}-\eqref{eq:estimate_energy} prove condition \eqref{eq:Phi_N}. Therefore, Lemma \ref{lemma:appendix_2} applies and
\begin{equation}\label{eq:gamma_b3_final}
|\eqref{eq:gamma_b3}|\;\leqslant\;{c}(t)\,(\alpha^<_N(t)+N^{-\eta'})\,
\end{equation}
for some ${c}(t)>0$.
The summand \eqref{eq:gamma_b4} is estimated straightforwardly as
\begin{equation}
\begin{split}\label{eq:gamma_b4_final}
|\eqref{eq:gamma_b4}|&\;\leqslant\;C\,N\|g_\beta(x_1-x_2)p_1\|_{\mathrm{op}}\big(\|u_t\|_\infty^2+\|v_t\|_\infty^2\big)\big(\|\widehat{m}^a\|_{\mathrm{op}}+\|\widehat{m}^b\|_{\mathrm{op}}\big)\\
&\;\leqslant\; \widetilde{c}(t) N^{-1-\beta/2+\xi},
\end{split}
\end{equation}
thanks to \eqref{eq:ma}, \eqref{eq:mb}, \eqref{eq:dressed2}, and \eqref{eq:g2},  where $C$ is a positive constant and $\widetilde{ c}(t)>0$ depends on $\|u_t\|_{\infty}$, $\|v_t\|_{\infty}$\,.

Choosing $\xi$ small enough, \eqref{eq:gamma_b3_final} and \eqref{eq:gamma_b4_final} yield
\[
|\eqref{eq:gamma_b1}|\;\leqslant\;\text {max}\{c(t),\widetilde{c}(t)\}\,(\alpha^<_N(t)+N^{-\eta''})
\]
for some $\eta''>0$, which, combined with \eqref{eq:gamma_b2_final}, again with $\xi$ small enough, finally gives
\[
|\delta_N^{(b)}(t)|\;\leqslant\; C(t)\,(\alpha^<_N(t)+N^{-\eta})
\]
for some $\eta>0$ and $C(t):=\text {max}\{\widetilde{C}(t),c(t), \widetilde{ c}(t)\}$.

\bigskip

\noindent\textbf{Term $\delta_N^{(c)}$.} The term
\[
\delta_N^{(c)}(t)\;=\;-4N(N-1)\mathfrak{Im}\langle\Psi_{N,t},\nabla_1g_\beta(x_1-x_2)\nabla_1R_{(12),t}\Psi_{N,t}\rangle_{\!\cH_N}
\]
has the very same structure as the term $\gamma_c$ discussed in \cite[page 38]{Pickl-RMP-2015}: therefore, re-doing the same computations therein we obtain
\[
|\delta_N^{(c)}(t)|\;\leqslant\; C(t)\,(\alpha^<_N(t)+N^{-\eta})
\]
for some $\eta>0$ and some $C(t)>0$ depending on $\|u_t\|_{H^2}$ and $\|v_t\|_{H^2}$ but not on $N$.

\bigskip

\noindent\textbf{Term $\delta_N^{(d)}$.} 	
Let us split
\begin{align}
&\delta_N^{(d)}(t)\;=\;N(N-1)(N-2)\mathfrak{Im}\langle\Psi_{N,t},g_\beta(x_1-x_2)\Big[V_N(x_1-x_3),R_{(12),t}\Big]\Psi_{N,t}\rangle_{\!\cH_N} \label{eq:gamma_d_1}\\
&+N(N-1)(N-2)\mathfrak{Im}\langle\Psi_{N,t},g_\beta(x_1-x_2)\Big[8\pi a(|u_t(x_3)|^2+|v_t(x_3|^2),R_{(12),t}\Big]\Psi_{N,t}\rangle_{\!\cH_N}\,. \label{eq:gamma_d_2}
\end{align}
Since the summand \eqref{eq:gamma_d_1} has the very same structure as the quantity $\gamma_d$ defined in \cite[Definition 6.3]{Pickl-RMP-2015}, one can merely repeat the analysis of \cite[Appendix A.2]{Pickl-RMP-2015} in order to bound it and to obtain
\begin{equation} \label{eq:gamma_d_final1}
|\eqref{eq:gamma_d_1}|\;\leqslant\; \widetilde{C}(t)\,(\alpha^<_N+N^{-\eta'})
\end{equation}
for some $\eta'>0$ and $\widetilde{C}(t)>0$ depending on $\|u_t\|_{\infty}$ and $\|v_t\|_{\infty}$ but not on $N$. In turn, we bound \eqref{eq:gamma_d_2} by means of \eqref{eq:commut_r2} in the form $[K_{34},R_{(12)}]$, and collecting all terms this produces
\[
\begin{split}
|&\eqref{eq:gamma_d_2}| \;\leqslant \;8\pi\,a\,N^3\big|\langle\Psi_{N,t},g_\beta(x_1-x_2)\big[|u_t(x_3)|^2+|v_t(x_3)|^2\,,\, (p_t)_1(p_t)_2(p_t)_3(p_t)_4\,\widehat{m}^c_t\big]\Psi_{N,t}\rangle_{\!\cH_N}\big|\\
+&8\pi\,a\,N^3\big|\langle\Psi_{N,t},g_\beta(x_1-x_2)\big[|u_t(x_3)|^2+|v_t(x_3)|^2\,,\,(p_t)_1(p_t)_2((p_t)_3(q_4)_t+(q_t)_3(p_t)_4)\,\widehat{m}^d_t\Big]\Psi_{N,t}\rangle_{\!\cH_N}\big|\\
+&8\pi\,a\,N^3\big|\langle\Psi_{N,t},g_\beta(x_1-x_2)\big[|u_t(x_3)|^2+|v_t(x_3)|^2\,,\,((p_t)_1(q_t)_2+(q_t)_1(p_t)_2)(p_t)_3(p_t)_4\,\widehat{m}^d_t\Big]\Psi_{N,t}\rangle_{\!\cH_N}\big|\\
+&8\pi\,a\,N^3\big|\langle\Psi_{N,t},g_\beta(x_1-x_2)\times\\
&\qquad\times\big[|u_t(x_3)|^2+|v_t(x_3)|^2\,,\,((p_t)_1(q_t)_2+(q_t)_1(p_t)_2)((p_t)_3(q_t)_4+(q_t)_3(p_t)_4)\,\widehat{m}^e_t\big]\Psi_{N,t}\rangle_{\!\cH_N}\big|\,.
\end{split}
\]
In each summand above there is at least one among $p_1$ and $p_2$ which we commute through $|u_t(x_3)|^2+|v_t(x_3)|^2$ until it hits $g_\beta$, and using $\|g_{12}p_1\|_{\mathrm{op}}=\|g_{12}p_2\|_{\mathrm{op}}$ we get
\[
|\eqref{eq:gamma_d_2}| \;\leqslant \;16\pi\,a\,N^3\|g_\beta(x_1-x_2)(p_t)_1\|_{\mathrm{op}}\,\, \big\||u_t|^2+|v_t|^2\big\|_\infty \big(\|\widehat{m}_t^c\|_{\mathrm{op}}+4\,\|\widehat{m}_t^d\|_{\mathrm{op}}+4\,\|\widehat{m}_t^e\|_{\mathrm{op}}\big)\,.
\]
Further, using the bounds \eqref{eq:mcde}, \eqref{eq:dressed2}, and \eqref{eq:g2}, we obtain
\begin{equation}\label{eq:gamma_d_final2}
|\eqref{eq:gamma_d_2}|\;\leqslant\; \widehat{C}(t)\,N^{-\beta/2+3\xi}
\end{equation}
for some $\widehat{C}(t)>0$ depending on $\|u_t\|_\infty$ and $\|v_t\|_\infty$ but not on $N$.
Finally, for $\xi$ small enough, \eqref{eq:gamma_d_final1} and \eqref{eq:gamma_d_final2} give
\[
|\delta_N^{(d)}(t)|\;\leqslant \;C(t)\,(\alpha^<_N+N^{-\eta})
\]
for some $\eta>0$ and for $C(t):=\text {max}\{\widetilde{C}(t),\widehat C(t)\}$.

\bigskip

\noindent\textbf{Term $\delta_N^{(e)}$.} Also for the term
\[
\delta_N^{(e)}(t)\;=\;\dfrac{1}{2}N(N-1)(N-2)(N-3)\mathfrak{Im}\langle\Psi_{N,t},g_\beta(x_1-x_2)\Big[V_N(x_3-x_4),R_{(12),t}\Big]\Psi_{N,t}\rangle_{\!\cH_N}
\]
we use \eqref{eq:commut_r2} in the form $[K_{34},R_{(12)}]$, and we get
\[
\begin{split}
|\delta_N^{(e)}(t)| &\;\leqslant \;\,N^4\,\big|\langle\Psi_{N,t},g_\beta(x_1-x_2)\big[V_N(x_3-x_4)\,,\, (p_t)_1(p_t)_2(p_t)_3(p_t)_4\,\widehat{m}^c_t\big]\Psi_{N,t}\rangle_{\!\cH_N}\big|\\
+&\,N^4\,\big|\langle\Psi_{N,t},g_\beta(x_1-x_2)\big[V_N(x_3-x_4)\,,\,(p_t)_1(p_t)_2((p_t)_3(q_t)_4+(q_t)_3(p_t)_4)\,\widehat{m}^d_t\big]\Psi_{N,t}\rangle_{\!\cH_N}\big|\\
+&\,N^4\,\big|\langle\Psi_{N,t},g_\beta(x_1-x_2)\big[V_N(x_3-x_4)\,,\,((p_t)_1(q_t)_2+(q_t)_1(p_t)_2)(p_t)_3(p_t)_4\,\widehat{m}^d_t\big]\Psi_{N,t}\rangle_{\!\cH_N}\big|\\
+&\,N^4\,\big|\langle\Psi_{N,t},g_\beta(x_1-x_2)\times\\
&\qquad\times\big[V_N(x_3-x_4)\,,\,((p_t)_1(q_t)_2+(q_t)_1(p_t)_2)((p_t)_3(q_t)_4+(q_t)_3(p_t)_4)\,\widehat{m}^e_t\big]\Psi_{N,t}\rangle_{\!\cH_N}\big|\,.
\end{split}
\]
As done for $\delta_N^{(d)}$, we commute one $p_1$ or $p_2$ through, until when it hits $g_\beta$. Moreover, we write $V_N=\mathbbm{1}_{\operatorname{supp}V_N}V_N$, and either $V_N(x_3-x_4)$ can be commuted through so as to multiply $\Psi_{N,t}$ in the left entry of the scalar product, or it already multiplies $\Psi_{N,t}$ in the right entry. The $\mathbbm{1}_{\operatorname{supp}V_N}$'s will then be used to provide extra $N$-decay. We thus find
\[
\begin{split}
|\delta_N^{(e)}(t)|\;\leqslant\;& C\,N^4\, \|g_\beta(x_1-x_2)(p_t)_1\|_{\mathrm{op}}\big(\|\mathbbm{1}_{\mathrm{supp}V_N}(x_3-x_4)(p_t)_3\|_{\mathrm{op}}+\|(p_t)_3\mathbbm{1}_{\mathrm{supp}V_N}(x_3-x_4)\|_{\mathrm{op}}\big)\\
&\quad\times\big(\|\widehat{m}^c_t\|_{\mathrm{op}}+\|\widehat{m}^d_t\|_{\mathrm{op}} + \|\widehat{m}^e_t\|_{\mathrm{op}} \big) \|V_N(x_3-x_4)\Psi_{N,t}\|\,,
\end{split}
\]
for some constant $C>0$.
We complete the estimate using the bounds \eqref{eq:mcde}, \eqref{eq:dressed2}, \eqref{eq:dressed3}, \eqref{eq:potential}, \eqref{eq:dressedpotential}, and \eqref{eq:g2}, together with $\|\mathbbm{1}_{\mathrm{supp}V_N}\|_2\;= \;C'\,N^{-3/2}$, and we get
\[
|\delta_N^{(e)}(t)|\;\leqslant\; C(t) N^{-\frac{\beta}{2}+3\xi}
\]
for $C(t)>0$ that depends on $\|u_t\|_{\infty}$ and $\|v_t\|_{\infty}$ but not on $N$. Taking $\xi$ small enough we obtain the desired estimate.

\bigskip

\noindent\textbf{Term $\delta_N^{(f)}$.} In
\[
\delta_N^{(f)}(t)\;=\;N(N-2)\mathfrak{Im}\langle\Psi_{N,t},g_\beta(x_1-x_2)\big[|u_t(x_1)|^2+|u_t(x_2)|^2+|v_t(x_1)|^2+|v_t(x_2)|^2,R_{(12),t}\big]\Psi_{N,t}\rangle_{\!\cH_N}
\]
the function $g_\beta$ can be always commuted so as to become adjacent to one of the $p$'s contained in $R_{(12),t}$. Thus, in the usual way, we get
\[
\begin{split}
|\delta_N^{(f)}(t)|&\;\leqslant\; \widetilde{C}(t)\, N^2\|g_\beta(x_1-x_2)(p_t)_1\|_{\mathrm{op}}\big(\|\widehat{m}^a_t\|_{\mathrm{op}}+\|\widehat{m}^b_t\|_{\mathrm{op}}\big)\\
&\;\leqslant\; \widehat{C}(t)\, N^{2}\|g_\beta\|_2\big(\|\widehat{m}^a_t\|_{\mathrm{op}}+\|\widehat{m}^b_t\|_{\mathrm{op}}\big)\\
&\;\leqslant\; C(t) N^{-\beta/2+\xi}
\end{split}
\]
for suitable functions $\widetilde{C}(t),\widehat{C}(t),C(t)>0$ depending on $\|u_t\|_{\infty}$ and $\|v_t\|_{\infty}$, having used \eqref{eq:ma}, \eqref{eq:mb}, \eqref{eq:dressed2}, and \eqref{eq:g2}. Taking $\xi$ small enough we obtain the desired estimate.
\end{proof}

%
%
%
%

\section*{Acknowledgements}
For this work both authors had the pleasure to benefit from instructive discussions with N.~Benedikter, S.~Cenatiempo, N.~Leopold, P.~T.~Nam, P.~Pickl, L.~Pitaevskii, M.~Porta, G.~Roati, C.~Saffirio, and S.~Stringari, which took place on the occasion of the their recent visits at SISSA, as well as from many exchanges with G.~Dell'Antonio and A.~Trombettoni at SISSA.
We also warmly acknowledge the kind hospitality of N.~Benedikter at the University of Copenhagen, of S.~Cenatiempo at the GSSI L'Aquila, and of  M.~Porta and B.~Schlein at the University of Zurich. This work is partially supported by the 2014-2017 MIUR-FIR grant ``\emph{Cond-Math: Condensed Matter and Mathematical Physics}'' code RBFR13WAET and by the mobility financial support of INdAM-GNFM.

\appendix
\section{A useful Lemma} 

The following result can be obtained by a suitable, actually straightforward adaptation of \cite[Lemma A.4]{Pickl-RMP-2015} to our two-component formalism.

This result is in fact crucial, as it connects the `projection counting' analysis for the derivation of the non-linear dynamics in the case of soft-scaling potential ($\beta<1$) to the actual Gross-Pitaevskii case of interest, $\beta=1$.

In our discussion this connection is needed when estimating the term $\delta_N^{(b)}$ in the course of the proof of Proposition \ref{prop:each_g_abcdef}.

\begin{lemma}\label{lemma:appendix_2}
	For fixed $a>0$ and $\beta\in(0,1)$, let $(\widetilde{W}_{\beta,N})_{N\in\mathbb{N}}$ be a sequence of spherically symmetric, positive, and compactly supported functions $\widetilde{W}_{\beta,N}\in L^\infty(\mathbb{R}^3,\mathbb{R})$ such that
	\begin{align}
	&\exists \;R>0\,\,\,\text {such that }, \,\,\forall N\in\mathbb{N},\,\text { one has }\, \widetilde{W}_{\beta,N}\;=\;0\,,\;\;\;\text {for}\,\,\, |x|\;\geqslant\; R \,N^{-\beta} \label{eq:Wtilde_1}\\
	& \lim_{N\to\infty}N^{1-3\beta}\|\widetilde{W}_{\beta,N}\|_\infty\;<\;+\infty \label{eq:Wtilde_2}\\
	&\exists\; \delta>0\,\,\,\text {such that }\,\,\,\lim_{N\to\infty}\,N^\delta\big|\|N\widetilde{W}_{\beta_N}\|_1-8\pi a\big|\;<\;+\infty\,. \label{eq:Wtilde_3}
	\end{align}
	Let a sequence $(\Phi_N)_{N\in\mathbb{N}}$ be such that $\Phi_N\in\cH_N$ and 
	\begin{equation} \label{eq:Phi_N}
	\sup_N\,\|\nabla_1\Phi_N\|_{\cH_N} \;\leqslant\;K
	\end{equation}
	for some $K>0$, and let $\alpha_{N,\Phi_N}^<$ be the functional $\alpha_N^<$ defined as in \eqref{eq:def_a<} relative to $\Phi_N$. Moreover, let $(u_t,v_t)$ be a solution to the system of non-linear Schr\"odinger equations \eqref{eq:GPsystem_extended} with the same constant $a$ considered here, and, correspondingly, let $R_{(12),t}$ be defined as in \eqref{eq:defnr}.
	Then,
	\begin{equation}
	\begin{split}
	N^2\Big|\big\langle \Phi_N,\Big[&\widetilde{W}_{\beta,N}-{\textstyle\frac{8\pi a}{N-1}}\big(|u_t(x_1)|^2+|v_t(x_1)|^2+|u_t(x_2)|^2+|v_t(x_2)|^2\big)\Big]R_{(12),t}\,\Phi_N\big\rangle_{\cH_N} \Big|\\
	&\;\leqslant\;D(t)(\alpha_{N,\Phi_N}^<(t)+N^{-\eta'})
	\end{split}	
	\end{equation}
	for some $\eta'>0$ and some positive function $D(t)$ depending on $\|u_t\|_{H^2}$, $\|v_t\|_{H^2}$, and K, but not on $N$.

\end{lemma}



\begin{thebibliography}{10}
	
	
	\bibitem{Aiba-Yajima-2013}
	D.~Aiba and K.~Yajima, {Schr{\"o}dinger equations with time-dependent strong
		magnetic fields}, {\em Algebra i Analiz} {\bf 25}(2) (2013)  37--62.
	
\bibitem{Anap-Hott-Hundertmark-2017}
I.~Anapolitanos, M.~Hott and D.~Hundertmark, {Derivation of the Hartree
equation for compound Bose gases in the mean field limit}, {\em Rev. Math. Phys.} {\bf 29}(1) (2017)  1750022, 28.
	
\bibitem{Antonelli-Weishaupl-2013}
P.~Antonelli and R.~M. Weish{\"a}upl, {Asymptotic behavior of nonlinear
{S}chr{\"o}dinger systems with linear coupling}, {\em J. Hyperbolic Differ.
Equ.} {\bf 11}(1) (2014)  159--183.

\bibitem{B-DO-S-gp-2015}
N.~Benedikter, G.~de~Oliveira and B.~Schlein, {Quantitative Derivation of the
Gross-Pitaevskii Equation}, {\em Communications on Pure and Applied
Mathematics} {\bf 68}(8) (2015)  1399--1482.
	
	\bibitem{Benedikter-Porta-Schlein-2015}
	N.~Benedikter, M.~Porta and B.~Schlein, {\em {Effective evolution equations
			from quantum dynamics}}, {Springer Briefs in Mathematical Physics}, Vol.~7
	(Springer, Cham, 2016).
	
	\bibitem{BS-2017}
	C.~Brennecke and B.~Schlein, {Gross-Pitaevskii Dynamics for Bose-Einstein
		Condensates} (https://arxiv.org/abs/1702.05625 (2017)).
	
	\bibitem{Bunoiu-Precup-2016}
	R.~Bunoiu and R.~Precup, {Vectorial approach to coupled nonlinear
		{S}chr{\"o}dinger systems under nonlocal {C}auchy conditions}, {\em Appl.
		Anal.} {\bf 95}(4) (2016)  731--747.
	
	\bibitem{EMS-2008}
	L.~{Erd\H os}, A.~Michelangeli and B.~Schlein, {Dynamical Formation of
		Correlations in a Bose-Einstein Condensate}, {\em Comm. Math. Phys.} {\bf 289}(3) (2009)  1171--1210.
	
	\bibitem{Hall2008_multicompBEC_experiments}
	D.~S. Hall, {\em {Multi-Component Condensates: Experiment}}, in {\em {Emergent
			Nonlinear Phenomena in Bose-Einstein Condensates: Theory and Experiment}\/},
	eds. P.~G. Kevrekidis, D.~J. Frantzeskakis and R.~Carretero-Gonz{\'a}lez.
	\newblock (Springer Berlin Heidelberg, Berlin, Heidelberg, 2008), pp. 307--327.
	
	\bibitem{HMEWC-1998}
	D.~S. Hall, M.~M. E., J.~R. Enscher, C.~E. Wieman and E.~A. Cornell, {The
		Dynamics of Component Separation in a Binary Mixture of Bose-Einstein
		Condensates}, {\em Phys. Rev. Lett.} {\bf 81}(8) (1998)  1539--1542.
	
\bibitem{Hall-Matthews-Wieman-Cornell_PRL81-1543}
D.~S. Hall, M.~R. Matthews, C.~E. Wieman and E.~A. Cornell, {Measurements of
Relative Phase in Two-Component Bose-Einstein Condensates}, {\em Phys. Rev.
Lett.} {\bf 81}(Aug 1998)  1543--1546.

\bibitem{Jeblick-Leopold-Pickl-2DGP-2016}
M.~{Jeblick}, N.~{Leopold} and P.~{Pickl}, {Derivation of the Time Dependent
Gross-Pitaevskii Equation in Two Dimensions} (arXiv.org:1608.05326 (2016)).

\bibitem{Jungel-Weishaupl2013_2compNLS_blowup}
A.~J{\"u}ngel and R.-M. Weish{\"a}upl, {Blow-up in two-component nonlinear
Schr{\"o}dinger systems with an external driven field}, {\em Mathematical
Models and Methods in Applied Sciences} {\bf 23}(09) (2013)  1699--1727.

\bibitem{kp-2009-cmp2010}
A.~Knowles and P.~Pickl, {Mean-field dynamics: singular potentials and rate of
convergence}, {\em Comm. Math. Phys.} {\bf 298}(1) (2010)  101--138.

\bibitem{LSeSY-ober}
E.~H. Lieb, R.~Seiringer, J.~P. Solovej and J.~Yngvason, {\em {The mathematics
	of the {B}ose gas and its condensation}}, {Oberwolfach Seminars}, Vol.~34
(Birkh{\"a}user Verlag, Basel, 2005).

\bibitem{Malomed2008_multicompBECtheory}
B.~Malomed, {\em {Multi-Component Bose-Einstein Condensates: Theory}}, in {\em
{Emergent Nonlinear Phenomena in Bose-Einstein Condensates: Theory and
	Experiment}\/},  eds. P.~G. Kevrekidis, D.~J. Frantzeskakis and
R.~Carretero-Gonz{\'a}lez.
\newblock (Springer Berlin Heidelberg, Berlin, Heidelberg, 2008), pp. 287--305.

\bibitem{MTCBM-PRL2004_BEC_heteronuclear}
M.~W. Mancini, G.~D. Telles, A.~R.~L. Caires, V.~S. Bagnato and L.~G. Marcassa,
{Observation of Ultracold Ground-State Heteronuclear Molecules}, {\em Phys.
Rev. Lett.} {\bf 92}(Apr 2004) p. 133203.
	
	\bibitem{Matthews_HJEWC_DMStringari_PRL1998}
	M.~R. Matthews, D.~S. Hall, D.~S. Jin, J.~R. Ensher, C.~E. Wieman, E.~A.
	Cornell, F.~Dalfovo, C.~Minniti and S.~Stringari, {Dynamical Response of a
		Bose-Einstein Condensate to a Discontinuous Change in Internal State}, {\em
		Phys. Rev. Lett.} {\bf 81}(Jul 1998)  243--247.
	
\bibitem{am_equivalentBEC}
A.~Michelangeli, {Equivalent definitions of asymptotic 100\% {BEC}}, {\em Nuovo
Cimento Sec.~B.}  (2008)  181--192.

\bibitem{am_GPlim}
A.~Michelangeli, {Role of scaling limits in the rigorous analysis of
{B}ose-{E}instein condensation}, {\em J. Math. Phys.} {\bf 48} (2007) p.
102102.

\bibitem{M-Olg-2016_2mixtureMF}
A.~Michelangeli and A.~Olgiati, {Mean-field quantum dynamics for a mixture of
Bose-Einstein condensates}, {\em Analysis and Mathematical Physics} {\bf 7}(4), 377-416 (2017).

	
\bibitem{Modugno-Ferrari-Inguscio-etal-Science2001_multicompBEC}
G.~Modugno, G.~Ferrari, G.~Roati, R.~J. Brecha, A.~Simoni and M.~Inguscio,
{Bose-Einstein Condensation of Potassium Atoms by Sympathetic Cooling}, {\em
Science} {\bf 294}(5545) (2001)  1320--1322.

\bibitem{Modugno-PRL-2002}
G.~Modugno, M.~Modugno, F.~Riboli, G.~Roati and M.~Inguscio, {Two Atomic
Species Superfluid}, {\em Phys. Rev. Lett.} {\bf 89}(Oct 2002) p. 190404.
	
\bibitem{MBGCW-1997}
C.~J. Myatt, E.~A. Burt, R.~W. Ghrist, E.~A. Cornell and C.~E. Wieman,
{Production of Two-overlapping Bose-Einstein Condensates by Sympathetic
Cooling}, {\em Phys. Rev. Lett.} {\bf 78}(4) (1997)  586--589.

\bibitem{AO-GPmixture-2016volume}
A.~Olgiati, {Effective Dynamics of Binary Condensates and Open Problems}, in
{\em {Advances in Quantum Mechanics: Contemporary Trends and Open
	Problems}\/},  eds. G.~Dell'Antonio and A.~Michelangeli, {\em {Springer INdAM
	Series}} (Springer International Publishing vol. 18 pp. 49--66).

\bibitem{AO-GP_magnetic_lapl-2016volume}
A.~Olgiati, {Remarks on the derivation of the Gross-Pitaevskii equation with
magnetic Laplacian}, in {\em {Advances in Quantum Mechanics: contemporary
trends and open problems}\/},  eds. G.~Dell'Antonio and A.~Michelangeli, {\em
{Springer INdAM Series}} (Springer International Publishing vol. 18 pp. 67--74).

\bibitem{Pachpatte-ineq}
B.~G. Pachpatte, {\em {Inequalities for differential and integral equations}},
{Mathematics in Science and Engineering}, vol.~197 (Academic Press, Inc., San
Diego, CA, 1998).
	
	\bibitem{Papp-Wieman_PRL2006_heteronuclear_RbRb}
	S.~B. Papp and C.~E. Wieman, {Observation of Heteronuclear Feshbach Molecules
		from a $^{85}\mathrm{Rb}$-$^{87}\mathrm{Rb}$ Gas}, {\em Phys. Rev.
		Lett.} {\bf 97}(Oct 2006) p. 180404.
	
\bibitem{Pickl-LMP-2011}
P.~Pickl, {A simple derivation of mean field limits for quantum systems}, {\em
Lett. Math. Phys.} {\bf 97}(2) (2011)  151--164.

\bibitem{Pickl-RMP-2015}
P.~Pickl, {Derivation of the time dependent {G}ross-{P}itaevskii equation with
external fields}, {\em Rev. Math. Phys.} {\bf 27}(1) (2015)  1550003, 45.
	
\bibitem{Pickl-JSP-2010}
P.~Pickl, {Derivation of the time dependent {G}ross-{P}itaevskii equation
without positivity condition on the interaction}, {\em J. Stat. Phys.} {\bf
140}(1) (2010)  76--89.
	
\bibitem{pita-stringa-2016}
L.~Pitaevskii and S.~Stringari, {\em {Bose-{E}instein {C}ondensation and
	{S}uperfluidity}} (Oxford University Press, 2016).

\bibitem{S-2008}
B.~Schlein, {Derivation of Effective Evolution Equations from Microscopic
Quantum Dynamics} (arXiv.org:0807.4307 (2008)).

\bibitem{S-2007}
B.~Schlein, {Dynamics of {B}ose-{E}instein Condensates} (arXiv.org:0704.0813
(2007)).

\bibitem{Jochen-Griesemer-2014}
J.~Schmid and M.~Griesemer, {Kato's theorem on the integration of
	non-autonomous linear evolution equations}, {\em Math. Phys. Anal. Geom.}
{\bf 17}(3-4) (2014)  265--271.
	
	\bibitem{Ketterle_StamperKurn_SpinorBEC_LesHouches2001}
	D.~M. Stamper-Kurn and W.~Ketterle, {\em {Spinor Condensates and Light
			Scattering from Bose-Einstein Condensates}}, in {\em {Coherent atomic matter
			waves: 27 July--27 August 1999}\/},  eds. R.~Kaiser, C.~Westbrook and
	F.~David.
	\newblock (Springer Berlin Heidelberg, Berlin, Heidelberg, 2001), Berlin,
	Heidelberg, pp. 139--217.
	
	\bibitem{StamperKurn-Ueda_SpinorBose_Gases_2012}
	D.~M. Stamper-Kurn and M.~Ueda, {Spinor Bose gases: Symmetries, magnetism, and
		quantum dynamics}, {\em Rev. Mod. Phys.} {\bf 85}(Jul 2013)  1191--1244.
	
\end{thebibliography}


\def\cprime{$'$}

\end{document}